\documentclass[12pt,english]{article}
\usepackage[T1]{fontenc}
\usepackage[latin9]{inputenc}
\usepackage{geometry}
\usepackage[active]{srcltx}
\usepackage{xcolor}
\usepackage{babel}
\usepackage{amsmath}
\usepackage{amsthm}
\usepackage{amssymb}
\usepackage{setspace}
\usepackage[authoryear]{natbib}
\usepackage[unicode=true,pdfusetitle,
 bookmarks=true,bookmarksnumbered=false,bookmarksopen=false,
 breaklinks=false,pdfborder={0 0 0},pdfborderstyle={},backref=false,colorlinks=true]
 {hyperref}
\hypersetup{
 citecolor=blue, linkcolor=blue, urlcolor=blue}

\makeatletter
\theoremstyle{definition}
\newtheorem{defn}{\protect\definitionname}
\theoremstyle{remark}
\newtheorem{claim}{\protect\claimname}
\theoremstyle{plain}
\newtheorem{thm}{\protect\theoremname}
\theoremstyle{remark}
\newtheorem{rem}{\protect\remarkname}
\theoremstyle{plain}
\newtheorem{cor}{\protect\corollaryname}
\theoremstyle{plain}
\newtheorem{fact}{\protect\factname}
\theoremstyle{plain}
\newtheorem{lem}{\protect\lemmaname}

 \usepackage{lmodern}
 \usepackage[T1]{fontenc}
\definecolor{purple}{RGB}{130, 40, 40}

\makeatother

\providecommand{\claimname}{Claim}
\providecommand{\corollaryname}{Corollary}
\providecommand{\definitionname}{Definition}
\providecommand{\factname}{Fact}
\providecommand{\lemmaname}{Lemma}
\providecommand{\remarkname}{Remark}
\providecommand{\theoremname}{Theorem}

\begin{document}
\title{Short-Term Investments and Indices of Risk}
\author{Yuval Heller\thanks{Department of Economics, Bar-Ilan University, Israel. yuval.heller@biu.ac.il.}
~and Amnon Schreiber\thanks{Corresponding author. Dept. of Economics, Bar-Ilan University, Israel.
amnon.schreiber@biu.ac.il.}\thanks{This paper replaces a working paper titled ``Instantaneous Decisions.''
We have benefited greatly from discussions with Sergiu Hart, Ilan
Kremer, Eyal Winter, and Pavel Chigansky, and from helpful comments
from seminar audiences at the Hebrew University of Jerusalem and Bar-Ilan
University. Yuval Heller is grateful to the European Research Council
for its financial support (starting grant \#677057). }}
\maketitle
\begin{abstract}
We study various decision problems regarding short-term investments
in risky assets whose returns evolve continuously in time. We show
that in each problem, all risk-averse decision makers have the same
(problem-dependent) ranking over short-term risky assets. Moreover,
in each problem, the ranking is represented by the same risk index
as in the case of CARA utility agents and normally distributed risky
assets.\textbf{}
\end{abstract}

\section{Introduction}

We study various decision problems regarding investments in risky
assets (henceforth, gambles), such as whether to accept a gamble,
or how to choose the optimal capital allocation. To rank the desirability
of gambles with respect to the relevant decision problem, it is often
helpful to use an objective riskiness index that is independent of
any specific subjective utility. For example, an objective riskiness
index is needed when pension funds are required not to exceed a stated
level of riskiness (see, e.g., the discussion in \citealp[p. 812]{Aumann_Serrano_2008}).

We analyze four decision problems that are important in economic settings.
In general, different risk-averse agents rank the desirability of
gambles differently. However, our main result shows that in each of
these problems, all risk-averse agents have the same (problem-dependent)
ranking over short-term investments in risky assets whose returns
evolve continuously. Moreover, in each problem, the ranking is represented
by the same risk index obtained in the commonly used mean-variance
preferences (e.g., \citealp{markowitz_52}), which are induced by
CARA utility agents and normally distributed gambles.

\paragraph{Brief Description of the Model}

We consider an agent who has to make an investment decision related
to a gamble. We think of a gamble as the additive return on a financial
investment. We assume that the agent has (1) an initial wealth $w$,
and (2) a von Neumann\textendash Morgenstern utility $u$ that is
increasing and risk-averse (i.e, $u'>0$ and $u''<0$). We assume
that a gamble is represented by a random variable with (1) positive
expectation, and (2) some negative values in its support. For each
problem the agents' choices are modeled by a \emph{decision function}
that assigns a number to each agent and each gamble, where a higher
number is interpreted as the agent finding the gamble to be more attractive
(i.e., less risky) for the relevant decision problem.

We study four decision problems in the paper: (1) \emph{acceptance/rejection},
in which the agent faces a binary choice between accepting and rejecting
the gamble (e.g., \citealp{Hart_2011}); (2) \emph{capital allocation},
in which the agent has a continuous choice of how much to invest in
the gamble (e.g., \citealp{markowitz_52,Sharpe_64}); (3) the \emph{optimal
certainty equivalent}, in which the agent evaluates how much an opportunity
to invest in the gamble (according to the optimal investment level)
is worth to the agent (e.g., \citealp{hellman2018indexing}); and
(4) \emph{risk premium, }in which the agent evaluates how much investing
in the gamble is inferior to obtaining the gamble's expected payoff
(\citealp{Arrow_71}).\footnote{\label{fn:We-use-the}We use ``risk premium'' in its common acceptation
in the economic literature since \citet{Arrow_71}. In the financial
literature (and in practice), the ``risk premium'' of a security
commonly has a somewhat different meaning, namely, the security return
less the risk-free interest rate (e.g., \citealp{cochrane2009asset}).}

A\emph{ risk index }is a function that assigns to each gamble a nonnegative
number, which is interpreted as the gamble's riskiness. We say that
a risk index is \emph{consistent} with a decision function $f$ over
some set of agents and gambles, if each agent in the set ranks all
gambles in the set according to that risk index; that is, $f$ assigns
for each agent a higher value for gamble $g$ than for gamble $g'$
iff the risk index assigns a lower value to $g$. A \emph{risk-aversion}
index is a function that assigns to each agent a non-negative number,
which is interpreted as the agent's risk aversion. We say that a risk-aversion
index is \emph{consistent} with a decision function over some set
of agents and gambles, if for each gamble and each pair of agents
in the set, the agent with the higher index of risk-aversion invests
less in the gamble than the other agent. Observe that different decision
functions may correspond to different concepts of risks, and, may
induce different indices of risk and of risk aversion.

\paragraph{Summary of Results}

Agents, typically, have heterogeneous rankings of gambles, and, thus,
no risk index (nor risk-aversion index) can be consistent with the
rankings of all agents, unless one restricts the set of gambles. Our
main result restricts the set of gambles to assets whose returns evolve
continuously in time, where the local uncertainty is induced by a
Wiener process. Specifically, we focus on Ito processes which are
continuous-time Markov processes. The class of Ito processes, is commonly
used in economic and financial applications and includes, in particular,
the geometric Brownian motion and mean-reverting processes (e.g.,
\citealp{merton1992continuous}).\footnote{Sec. \ref{subsec:Continuous-Time-Processes-with} demonstrates that
our results \emph{cannot} be extended to continuous-time processes
with jumps.}

Our main result shows that in each of the four decision problems discussed
above, all agents rank all gambles in the same (problem-dependent)
way when they have to decide on short-term investments in gambles
whose returns evolve continuously in time. Moreover, the risk indices
that are consistent with these decision functions are the same as
in the classic model of agents with CARA (exponential) utilities and
normally distributed gambles. Specifically, we show that: (1) the
variance-to-mean index $Q_{VM}\left(g\right)=\frac{\sigma^{2}\left[g\right]}{\boldsymbol{E}\left[g\right]}$
is consistent with both the capital allocation function and the acceptance/rejection
function, (2) the inverse Sharpe index $Q_{IS}\left(g\right)=\frac{\sigma\left[g\right]}{\boldsymbol{E}\left[g\right]}$
is consistent with the optimal certainty equivalent function, and
(3) the standard deviation index $Q_{SD}\left(g\right)=\sigma\left[g\right]$
is consistent with the risk premium function. Finally, we adapt the
classic results of \citet{Pratt_64} and \citet{Arrow_71} to the
present setup, and show that the local Arrow\textendash Pratt coefficient
of absolute risk aversion $\rho\left(u,w\right)=\frac{-u''\left(w\right)}{u'\left(w\right)}$
is consistent with all four decision functions.

\paragraph{Related Literature and Contribution}

\citet{Aumann_Serrano_2008} and \citet{Foster_Hart_2009} presented
two ``objective'' indices of riskiness of gambles, which are independent
of the subjective utility of the agent. These indices are either based
on reasonable axioms that an index of risk should satisfy (e.g., \citealp{Artzner_at_al_99,Aumann_Serrano_2008,cherny_madan_2009,Foster_Hart_2011,schreiber_relative_2014,hellman2018indexing};
see also the recent survey of \citealp{follmer2015axiomatic}), or
they are based on an ``operative'' criterion such as an agent never
going bankrupt when relying on an index of risk in deciding whether
to accept a gamble (\citealp{Foster_Hart_2009}; and see also \citealp{meilijson_2009},
for a discussion of operative implication of \citeauthor{Aumann_Serrano_2008}'s
index of risk).\footnote{\citeauthor{Aumann_Serrano_2008}'s \citeyearpar{Aumann_Serrano_2008}
and \citeauthor{Foster_Hart_2009}'s \citeyearpar{Foster_Hart_2009}
indices of risk have been extended to gambles with an infinite support
(\citealp{homm2012operational,schulze2014existence,riedel2015foster})
and to gambles with unknown probabilities (\citealp{michaeli2014riskiness}).
These indices have been applied to study real-life investment strategies
in \citet{kadan2014performance,bali2015new,anand2016foster,leiss2018option}.}

We argue that risk is a multidimensional attribute that crucially
depends on the investment problem. Different aspects of risk are relevant
when an agent has to decide whether to accept a gamble, compared with
a situation in which an agent has to choose how much to invest in
a gamble, or has to evaluate the certainty equivalent of the optimal
investment. Many existing papers focus on a single decision function.
By contrast, we suggest a framework for studying various decision
problems, and associate each such problem with its relevant index
of risk. We believe that this general framework may be helpful in
future research on risk indices.

In general, different agents make different investment decisions,
based on the subjective utility of each agent. Thus, a single risk
index cannot be consistent with the choices of all agents, which,
arguably, limits the index's objectiveness (even when the index satisfies
appealing axioms or some operative criterion for avoiding bankruptcy).
However, our main result shows that in various important decision
problems, all agents rank all gambles in the same way when deciding
on short-term investments in gambles whose returns evolve continuously
in time. This finding enables us to construct objective risk indices
that are consistent with the short-term investment decisions of risk-averse
agents.

There are pairs of gambles for which all risk-averse agents agree
on which one of the gambles is more desirable. This happens if one
gamble second-order stochastically dominates the other gamble. However,
the well-known order of stochastic dominance (\citealp{Hadar_Russell_69,Hanoch_Levy_69,rothschild_stiglitz_1})
is only a partial order and ``most'' pairs of gambles are incomparable.
Interestingly, even if one gamble second-order stochastically dominates
another gamble, it is not sufficient for a uniform ranking among all
risk-averse agents in every decision problem (see, e.g., the analysis
of capital allocation decisions in \citealp{landsberger1993mean}).

A large body of literature uses the classic mean-variance capital
asset pricing model (\citealp{markowitz_52}; see \citealp{smetters2013sharper,kadan2016Liu}
for recent extensions). A well-known critique is that in a discrete-time
setup the mean-variance preferences are consistent with expected utility
maximization only under severe restrictions, such as CARA utilities
and normally distributed gambles (see, e.g., \citealp{borch1969note,feldstein1969mean,hakansson1971multi}).
By contrast, the seminal results of Robert Merton (as summarized and
discussed in \citealp{merton1975theory,merton1992continuous}) show
that in a continuous-time model with log-normally distributed asset
prices mean-variance preferences are consistent with the optimal portfolio
allocation of all risk-averse agents. Merton's results present an
important theoretical foundation for the classic model.

The present paper extends this idea by showing that the equivalence
between the decisions of agents with CARA utilities with respect to
normally distributed gambles and the decisions of risk-averse agents
with respect to short-term investments holds more broadly: (1) it
holds with respect to various decision functions beyond optimal portfolio
allocation, and (2) it holds with respect to a broad class of continuous-time
processes beyond log-normally distributed asset prices. On the other
hand, our analysis is less general than \citeauthor{merton1992continuous}'s
in that we analyze situations in which the agent acts only at the
beginning, at time zero, and cares about his wealth at a single future
nearby point $t>0$, rather than allowing the agent to act continuously
in time.

Our paper is also related to the literature on local risks. This literature
focuses on discrete-time gambles, rather than continuous-time returns
(which are the focus of the present paper). \citet{Pratt_64} shows
that if the distribution of the returns is sufficiently concentrated,
i.e., the third absolute central moment is sufficiently small relative
to variance, then for any agent, the magnitude of the risk premium
depends on the local level of the agent's risk aversion. \citet{Samuelson_70}
shows that classic mean-variance analysis (\citealp{markowitz_52}),
applies approximately to all utility functions in situations that
involve what he calls ``compact'' distributions. More recently,
\citet{schreiber2015note} shows that if one gamble is riskier than
another gamble according to the \citeauthor{Aumann_Serrano_2008}'s
index of risk, then every decision maker who is willing to accept
a small proportion of the riskier gamble is also willing to accept
the same proportion of the less risky gamble.

In this context, two papers are close to the present paper: \citet{shorrer_1}
and \citet{schreiber_acceptance_2013}. \citeauthor{shorrer_1} shows
that there exist risk indices that are consistent with the acceptance/rejection
decisions of all risk-averse agents with respect to bounded discrete
gambles with sufficiently small support. This result is similar to
our characterization of risk indices that are consistent with various
short-term investment decisions of all risk-averse agents with respect
to assets whose (possibly, unbounded) returns evolve continuously
in time. \citeauthor{shorrer_1}'s main result shows that by adding
a few additional axioms, one can uniquely choose \citeauthor{Aumann_Serrano_2008}'s
index among all the indices that are consistent with agents' acceptance/rejection
with respect to small discrete gambles.\footnote{\citet{shorrer_1} further applies analogous axioms in the related
setup in which an agent has to accept/reject an option to allocate
a certain amount of money in a multiplicative gamble, and other interesting
setups  that deal with acceptance/rejection of  cash flows and information
transactions.} In principle, one could apply a similar axiomatic method to our three
other decision functions; we leave this interesting research direction
for future research (for further discussion see Sec. \ref{sec:Discussion}).
Unlike the other papers mentioned above, \citet{schreiber_acceptance_2013}
deals with returns in the continuous-time setup. Specifically, he
analyzes acceptance and rejection of short-term investments. The key
contributions of the present paper with respect to \citet{schreiber_acceptance_2013}
consists in, first, extending the analysis to the other three decision
functions (namely, capital allocation, optimal certainty equivalent,
and risk premium) and, second, showing in all four cases an equivalence
to the indices in the exponential-normal setup.

\paragraph{Structure}

In Section \ref{sec:Model} we present our model. In Section \ref{sec:Normal-Distributions-=000026}
we analyze the benchmark setup of CARA utilities and normally distributed
gambles. In Section \ref{sec:Continuous-Investment-Decisions} we
adapt the model to study risky assets whose returns evolve continuously
in time, and present our main result. We conclude with a discussion
in Section \ref{sec:Discussion}. Appendix \ref{sec:Multiplicative-Gambles}
extends our model to multiplicative gambles. The formal proofs are
presented in Appendix \ref{sec:Proofs}.

\section{Model\label{sec:Model}}

We consider an agent who has to make an investment decision related
to a risky asset. We begin by defining each of these components: agent,
risky asset, and investment decision.

A \emph{decision maker} (or \emph{agent}) is modeled as a pair $\left(u,w\right)$,
where $u:\mathbb{R}\rightarrow\mathbb{R}$ is a twice continuously
differentiable von Neumann\textendash Morgenstern utility function
over wealth satisfying $u'>0$ (i.e., utility is increasing in wealth)
and $u''<0$ (i.e., risk aversion), and $w\in\mathbb{R}$ is an initial
wealth level. Let $\mathcal{DM}$ denote the set of all such decision
makers.

A \emph{gamble} $g$ is a real-valued random variable with a positive
expectation and some negative values (i.e., $0<\boldsymbol{E}\left[g\right]$,
and $\boldsymbol{P}\left[g<0\right]>0$). We think of a gamble as
the additive return on a risky investment; for example, if the initial
investment is $x$ dollars and the random payoff from the investment
is $y$ dollars, then the additive return $g\equiv y-x$ is a gamble.
Let $\mathcal{G}$ denote the set of all such gambles.

A \emph{decision function }$f:\mathcal{DM}\times\mathcal{G}\rightarrow\mathbb{R}$
is a function that assigns to each agent and each gamble a nonnegative
number, where a higher value is interpreted as the agent finding the
gamble to be more attractive (i.e., less risky) for the relevant investment
decision.

\subsection{Decision Functions}

We study four decision functions in the paper:
\begin{enumerate}
\item \emph{Acceptance/rejection}: We consider a situation in which an agent
faces a binary choice between accepting and rejecting the gamble.
Specifically, the acceptance function $f_{AR}:\mathcal{DM}\times\mathcal{G}\rightarrow\left\{ 0,1\right\} $
is given by
\[
f_{AR}\left(\left(u,w\right),g\right)=\begin{cases}
1 & \mathbb{E}\left[u\left(w+g\right)\right]\geq u\left({\color{purple}w}\right)\\
0 & \mathbb{E}\left[u\left(w+g\right)\right]<u\left({\color{purple}w}\right).
\end{cases}
\]
That is, $f_{AR}\left(\left(u,w\right),g\right)$ is equal to one
if accepting the gamble yields a weakly higher expected payoff than
rejecting it, and it is equal to zero otherwise. The acceptance function
has been used to study risk indices in various papers (e.g., \citealp{Foster_Hart_2009,Foster_Hart_2011}).
In particular, our analysis of this decision function extends the
analysis of \citet{schreiber_acceptance_2013}, by showing the similarity
between this function in the mean-variance setup and the corresponding
decision function in the continuous-time setup.
\item \emph{Capital allocation}: Second, we study a situation in which an
agent has a continuous choice of how much to invest in the gamble.
Specifically, the capital (or asset) allocation function $f_{CA}:\mathcal{DM}\times\mathcal{G}\rightarrow\mathbb{R}^{+}\cup\left\{ \infty\right\} $
is given by 
\begin{equation}
f_{CA}\left(\left(u,w\right),g\right)=\arg\max_{\alpha\in\mathbb{R}^{+}}\mathbb{E}\left[u\big(w+\alpha g\big)\right];\label{eq:f_CA}
\end{equation}
if (\ref{eq:f_CA}) does not admit of a solution (i.e., $\mathbb{E}\left[u\big(w+\alpha g\big)\right]$
is increasing for all $\alpha$-s), then we set $f_{CA}\left(\left(u,w\right),g\right)=\infty$.
That is, $f_{CA}\left(\left(u,w\right),g\right)$ is the optimal level
the agent $\left(u,w\right)$ chooses to invest in gamble $g$. An
investment level of zero is interpreted as no investment in the gamble.
An investment level in the interval $\left(0,1\right)$ is interpreted
as a partial investment in the gamble. An investment level of one
is interpreted as a total investment in the gamble (without leverage).
Finally, an investment level strictly greater than one is interpreted
as a more than total investment in the gamble (achieved, for example,
through high leverage). The capital allocation function is prominent
in classic analyses of riskiness of assets (e.g., \citealp{markowitz_52,Sharpe_64}),
and, more recently, it has been used to derive an \emph{incomplete}
ranking over the riskiness of gambles (\citealp{landsberger1993mean}).
\item \emph{The optimal certainty equivalent: }Third, we study a situation
in which an agent has to assess how much an opportunity to invest
in the gamble $g$ is worth to him (where we allow the agent to choose
his optimal investment level). Specifically, the optimal certainty
equivalent function $f_{CE}:\mathcal{DM}\times\mathcal{G}\rightarrow\mathbb{R}^{+}\cup\left\{ \infty\right\} $
is defined implicitly as the unique solution to the equation
\begin{equation}
u\left(w+f_{CE}\right)=\max_{\alpha\in\mathbb{R}^{+}}\mathbb{E}\left[u\big(w+\alpha{\color{purple}\cdot}g\big)\right];\label{eq:f_CE}
\end{equation}
if (\ref{eq:f_CE}) does not admit of a solution (which happens when
$\mathbb{E}\left[u\big(w+\alpha{\color{purple}\cdot}g\big)\right]$
is increasing for all $\alpha$-s), then we set $f_{CE}=\infty$.
That is, $f_{CE}\left(\left(u,w\right),g\right)$ is interpreted as
the certain gain for which the decision maker is indifferent between
obtaining this gain for sure and having an option to invest in the
gamble $g$, when the agent is allowed to optimally choose his investment
level in $g$. Observe that one can express the RHS in (\ref{eq:f_CE})
in terms of $f_{CA}$ and obtain the following equivalent definition
of $f_{CE}$ as the unique solution to the equation $u\left(w+f_{CE}\right)=\mathbb{E}\left[u\big(w+f_{CA}\left(\left(u,w\right),g\right)\cdot g\big)\right]$.
The function $f_{CE}$ has been studied axiomatically in \citet{hellman2018indexing}.
\item \emph{Risk premium}: Lastly, we study a situation in which the agent
has to decide between investing in the gamble and obtaining a certain
amount that is less than the gamble's expected payoff. Specifically,
the risk premium function $f_{RP}:\mathcal{DM}\times\mathcal{G}\rightarrow\mathbb{R}^{-}\cup\left\{ -\infty\right\} $
is defined implicitly as the unique solution to the equation
\begin{equation}
\mathbb{E}\left[u\left(w+g\right)\right]=u\left(w+E\left[g\right]+f_{RP}\right);\label{eq:f_RP}
\end{equation}
if such a solution does not exist then we set $f_{RP}\left(\left(u,w\right),g\right)=-\infty$
. That is, $f_{RP}\left(\left(u,w\right),g\right)$ is interpreted
as the negative amount that has to be added to the expected value
of the gamble, to make the agent indifferent between investing in
the gamble, and obtaining the gamble's expected payoff plus this negative
amount. Here we use the common acceptation of risk premium in the
economic literature (\citealp{Arrow_71}; see \citealp[Section 3.2]{kreps1990course},
for a textbook definition), which has a somewhat different meaning
in some of the finance literature (see Footnote \ref{fn:We-use-the}).
\end{enumerate}
In the main text we study additive gambles, in which the gamble's
realized outcome is added to the initial wealth. In Appendix \ref{sec:Multiplicative-Gambles}
we extend our model to multiplicative gambles, in which the realized
outcome of the gamble is interpreted as the per-dollar return.

\subsection{Risk Indices\label{subsec:Risk-Indexes}}

We define a \emph{risk index }as a function $Q:\mathcal{G}\rightarrow\mathbb{R^{++}}$
that assigns to each gamble a positive number, which is interpreted
as the gamble's riskiness. We study three risk indices: 
\begin{enumerate}
\item The \emph{variance-to-mean index} $Q_{VM}\left(g\right)$ is the ratio
of the variance to the mean:
\[
Q_{VM}\left(g\right)=\frac{\sigma^{2}\left[g\right]}{\boldsymbol{E}\left[g\right]},\,\,\,\,\,\,\,\,\,\,\,\,\textrm{where}\,\,\sigma^{2}\left[g\right]\equiv E\left[\left(g-E\left[g\right]\right)^{2}\right].
\]
\item The\emph{ inverse Sharpe index} $Q_{IS}\left(g\right)$ is the ratio
of the standard deviation to the mean:
\[
Q_{IS}\left(g\right)=\frac{\sigma\left[g\right]}{\boldsymbol{E}\left[g\right]}.
\]
\item The \emph{standard deviation} \emph{index} $Q_{SD}\left(g\right)$
is equal to: $Q_{SD}\left(g\right)=\sigma\left[g\right].$
\end{enumerate}
We say that a risk index is consistent with a decision function over
a domain of agents and gambles, if: each agent in the domain finds
gamble $g$ less attractive than $g'$ with respect to the relevant
decision function iff the risk index of $g$ is higher than in $g'$.
Formally:
\begin{defn}
\label{def:Risk-index-}Risk index $Q$ is \emph{consistent} with
$f$ over the domain $DM\times G\subseteq\mathcal{DM}\times\mathcal{G}$
if 
\[
Q\left(g\right)>Q\left(g'\right)\Leftrightarrow f\left(\left(u,w\right),g\right)<f\left(\left(u,w\right),g'\right)
\]
 for each agent $\left(u,w\right)\in DM$ and each pair of gambles
$g,g'\in G$.

Our definition of consistency is restrictive and for a given domain
of gambles and agents it may not apply at all. In particular observe
that a domain $DM\times G\subseteq\mathcal{DM}\times\mathcal{G}$
admits a consistent risk index iff \emph{all agents have the same
ranking over gambles}, i.e., if 
\[
f\left(\left(u,w\right),g\right)<f\left(\left(u,w\right),g'\right)\Leftrightarrow f\left(\left(u',w'\right),g\right)<f\left(\left(u',w'\right),g'\right)
\]
 for each pair of agents $\left(u,w\right),\left(u',w'\right)\in DM$
and each pair of gambles $g,g'\in G$.

Note that consistency is an ordinal concept; i.e., a consistent risk
index is unique up to monotone transformations; if risk index $Q$
is \emph{consistent} with function $f$ over the domain $DM\times G$,
then risk index $Q'$ is consistent with $f$ over this domain iff
there exists a strictly increasing mapping $\theta:$ $Q\left(G\right)\rightarrow Q'\left(G\right)$,
s.t. $Q'\left(g\right)=\theta\left(Q\left(g\right)\right)$ for each
gamble $g\in G$.
\end{defn}

\subsection{Risk-Aversion Indices}

We define a \emph{risk-aversion index }as a function $\phi:\mathcal{DM}\rightarrow\mathbb{R^{++}}$
that assigns to each agent a non-negative number, which is interpreted
as the agent's risk aversion. We mainly study one risk index in the
paper, the Arrow\textendash Pratt coefficient of absolute risk aversion,
denoted by $\rho:\mathcal{DM}\rightarrow\mathbb{R^{++}}$, which is
defined as follows: 
\[
\rho\left(u,w\right)=\frac{-u''\left(w\right)}{u'\left(w\right)}.
\]

We say that a risk-aversion index is consistent with a decision function
over a domain of agents and gambles, if, for each gamble and each
pair of agents in the domain, the agent with the higher index chooses
a lower value for his investment decision in the gamble. 
\begin{defn}
\label{def:Risk-aversion-index-1}Risk-aversion index $\phi$ is \emph{consistent}
with $f$ over $DM\times G\subseteq\mathcal{DM}\times\mathcal{G}$
if 
\[
\phi\left(u,w\right)>\phi\left(u',w'\right)\Leftrightarrow f\left(\left(u,w\right),g\right)<f\left(\left(u',w'\right),g\right)
\]
 for each pair of agents $\left(u,w\right),\left(u',w'\right)\in DM$
and each gamble $g\in G$.

Here again, the definition of consistency is restrictive and for a
given domain of gambles and agents it may not apply at all. Specifically,
a domain $DM\times G\subseteq\mathcal{DM}\times\mathcal{G}$ admits
a consistent risk-aversion index iff all \emph{gambles induce the
same ranking over agents}, i.e., if 
\[
f\left(\left(u,w\right),g\right)<f\left(\left(u',w'\right),g\right)\Leftrightarrow f\left(\left(u,w\right),g'\right)<f\left(\left(u',w'\right),g'\right)
\]
for each pair of agents $\left(u,w\right),\left(u',w'\right)\in DM$
and each pair of gambles $g,g'\in G$. Further, the consistency of
a risk-aversion index is unique up to a strictly monotone transformation.
\end{defn}

\section{Normal Distributions and CARA Utilities\label{sec:Normal-Distributions-=000026}}

\subsection{Result}

We begin by presenting a claim, which summarizes known results for
normal distributions and CARA utilities. Specifically, we show that
in each of the decision functions described above, all agents with
CARA utilities have the same ranking over all normally distributed
gambles, and that each of these rankings is consistent with one of
the risk indices presented above. Moreover, all normally distributed
gambles induce the same ranking over all agents with CARA utilities,
which is consistent with the Arrow\textendash Pratt coefficient.

Formally, let $DM_{CARA}\subseteq\mathcal{DM}$ be the set of decision
makers with CARA utilities:\footnote{Clearly, one can extend the definition of $DM_{CARA}$ (without affecting
any of the results) by allowing the utilities to differ from $1-e^{-{\color{purple}\rho}\cdot x}$
by adding a constant and multiplying by a positive scalar.} 
\[
DM_{CARA}=\left\{ \left(u,w\right)\in DM\,|\exists{\color{purple}\rho}>0,\,\,s.t.\,\,u\left(x\right)=1-e^{-{\color{purple}\rho}\cdot x}\right\} ,
\]
and let $\mathcal{G}_{N}\subseteq\mathcal{G}$ be the set of normally
distributed gambles with positive expectations:
\[
G_{N}=\left\{ g\in\mathcal{G}|g\sim Norm\left(\mu,\sigma\right),\,\,\textrm{for some }\mu,\sigma>0\right\} .
\]

\begin{claim}
\label{claim:normal-exponent}Let $u$ be a CARA utility with parameter
$\rho$ (i.e., $u\left(x\right)={\color{purple}1-e^{-{\color{purple}\rho}\cdot x}}$).
Then:
\begin{enumerate}
\item $f_{RP}((u,w),g)=-0.5\cdot\rho\cdot\sigma^{2}$, which implies that
the standard deviation index $Q_{SD}$ is consistent with the risk
premium function $f_{RP}$ in the domain $DM_{CARA}\times G_{N}$.
\item (I)$f_{AR}((u,w),g)=1$ iff $\frac{2}{\rho}\cdot\frac{\mu}{\sigma^{2}}\geq1$
(and $f_{AR}((u,w),g)=0$ otherwise), (II) $f_{CA}$$((u,w),g)=\frac{1}{\rho}\cdot\frac{\mu}{\sigma^{2}}$,
which imply that the variance-to-mean index $Q_{VM}$ is consistent
with both the acceptance/rejection function $f_{AR}$ and the capital
allocation function $f_{CA}$.
\item $f_{CE}((u,w),g)=\frac{1}{2\cdot\rho}\cdot\left(\frac{\mu}{\sigma}\right)^{2}$,
which implies that the inverse Sharpe index $Q_{IS}$ is consistent
with the optimal certainty equivalent function $f_{CE}$.
\item The Arrow\textendash Pratt coefficient of absolute risk aversion $\rho$
is consistent with all four decision functions $f_{CA}$, $f_{AR}$,
$f_{CE}$, and $f_{RP}$ in the domain $DM_{CARA}\times G_{N}$.
\end{enumerate}
\end{claim}
For completeness, we present the proof of Claim \ref{claim:normal-exponent}
in Appendix \ref{subsec:Proof-of-Claim}.

\subsection{Discussion}

Each agent with CARA utility is described by two parameters (initial
wealth $w$, and Arrow\textendash Pratt coefficient $\rho$). Similarly,
each normal gamble is described by two parameters (expectation $\mu$
and standard deviation $\sigma$). This implies that any decision
function can be expressed as a function $g\left(w,\rho,\mu,\sigma\right)$
of these four parameters.

\paragraph{Other Consistent Indices of Risk Aversion}

CARA utilities have the well-known property that the initial wealth
does not affect expected utility calculations with respect to investments
in gambles. Thus, whenever the investment decision is made by choosing
the option that maximizes the agent's expected utility (such as in
all four of the decision functions analyzed above), then the decision
function is independent of $w$, which implies that the parameter
$\rho$ is a consistent risk-aversion index. By contrast, for investment
decisions that are not determined by maximizing the agent's expected
utility, there might be different risk-aversion consistent indices.
For instance, \citet{Foster_Hart_2009} analyze a situation in which
an agent accepts or rejects gambles while his goal is to avoid bankruptcy.
The index of risk aversion that is consistent with their decision
function is the wealth level.

\paragraph{Separability Condition for Having a Consistent Risk Index}

The decision functions analyzed above have the additional separability
property that each function $f$ can be represented as a product of
two functions: one that depends only on the parameters describing
the agent ($w$ and $\rho$), and one that depends only on the parameters
of the gamble, i.e., $f\left(\left(u,w\right),g\right)=\tilde{f}\left(w,\rho,\sigma,\mu\right)=h\left(w,\rho\right)\cdot\nu\left(\mu,\sigma\right)$.
This separability implies that all agents with CARA utilities have
the same ranking over normal gambles (as this ranking depends only
on $\nu\left(\mu,\sigma\right)$, which does not depend on the agent's
parameters), which, in turn, implies that there exists a consistent
risk index. Similarly, the separability implies that all normal gambles
induce the same ranking over agents (as this ranking depends only
on $h\left(w,\rho\right)$, which does not depend on the parameters
of the normal gamble).

Other decision functions might not satisfy this separability property.
One example of such a non-separable decision function is \textit{\emph{the
standard certainty equivalent}} of a continuous gamble $f_{SCE}$
(as opposed to the certainty equivalent of the optimal allocation
of the gamble $f_{CE}\left(\left(u,w\right),g\right)$ discussed above),
which is implicitly defined by 
\[
\mathbb{E}\left[u(w+g)\right]=u(w+f_{SCE}).
\]

The definitions of $f_{SCE}$ and $f_{RP}$ imply that $f_{SCE}=\mathbb{E}\left[g\right]+f_{RP}.$
Substituting $f_{RP}=-0.5\cdot\rho\cdot\sigma^{2}$ (which is proven
in Appx. \ref{subsec:Proof-of-Claim}) yields: $f_{SCE}=\mathbb{E}\left[g\right]+f_{RP}=\mu-0.5\cdot\rho\cdot\sigma^{2},$
which is a non-separable function of $\rho$,$\mu,\sigma$. The non-separability
implies that agents with different CARA utilities have different rankings
for normal gambles and, therefore, no risk index can be consistent
with these decisions.

\section{Short-Term Investments in Continuous Gambles\label{sec:Continuous-Investment-Decisions}}

In what follows we adapt our model to short-term investment decisions
regarding assets whose value follows a continuous random process.
Our description of the continuous-time setup follows \citet{shreve2004stochastic}.

\subsection{Continuous-Time Random Processes}

Let $\left(\Omega,\mathcal{F},{\normalcolor \mathbb{P}}\right)$ be
a probability space on which a Brownian motion $W_{t}$ is defined,\footnote{For ease of exposition we limit the Wiener process to one dimension.
All the results remain the same with a multidimensional Wiener process
with the corresponding adjustments.} with an associated filtration $\mathcal{F}(t)$. Let the process
$g$ be described by the following stochastic differential equation:

\begin{equation}
dg_{t}=\mu_{t}dt+\sigma_{t}dW_{t},\label{eq:Wiener}
\end{equation}
where the drift $\mu$ and the diffusion $\sigma$ are adapted stochastic
processes (i.e., $\mu_{t}$ and $\sigma_{t}$ are $\mathcal{F}(t)$-measurable
for each $t$, see \citealp[Footnote 1 on page 97]{shreve2004stochastic})
and are both continuous in $t$. We assume that $\mu_{0}>0$ and $\sigma_{0}>0$
and that $\sigma_{t}\geq0$ for all $t>0$. We also assume that $E\int_{0}^{t}\sigma_{s}^{2}ds$
and $E\int_{0}^{t}|\mu_{s}|ds$ are finite for every $t>0$. This
implies that the integrals $E\int_{0}^{t}\sigma_{s}dW_{s}$ and $E\int_{0}^{t}\mu ds$
are well defined, and the Ito integral $\int_{0}^{t}\sigma_{s}dW_{s}$
is a martingale; see \citet[Footnote 2 on page 143]{shreve2004stochastic}.

The process $g$ is interpreted as the additive return of some risky
asset. Specifically, let $P_{t}$ be the price of some risky asset
at time $t$ and assume that $P_{0}$ is known. Then, 
\[
g_{t}=P_{t}-P_{0},
\]
 is the additive return of the asset at time $t$. In particular,
observe that $g_{0}=P_{t=0}-P_{0}\equiv0$. For simplicity we assume
that $g$ is bounded from below (i.e., there exists $M_{g}\in\mathbb{R}$,
such that $g_{t}\geq M_{g}$ for each $t>0$). Under these assumptions,
for a sufficiently small time $t$, $g_{t}$ is a gamble (see Footnote
\ref{fn:The-analysis-implies} in Appendix \ref{subsec:Proof-of-Theorem});
i.e., it has positive expectation and takes negative values with positive
probability. Thus, we can apply our definitions of decision functions
and indices to $g_{t}$ for each specific value of $t>0$.

In our setup, a decision maker has to make a decision at time zero,
where he cares only about his wealth at some time t. From the perspective
of this decision maker, $P_{0}$ is a pure number and $g_{t}$ is
a gamble, just as before.

Let $\Gamma$ denote the set of all continuous-time random processes
satisfying these assumptions. Observe that the set $\Gamma$ is quite
general. In particular, it includes returns on assets whose prices
are described by geometric Brownian processes\emph{ }(\citealp{black1973pricing}\emph{;
}\citealp[Chapters 4 and 5]{merton1992continuous})\emph{,}\footnote{When one models an asset's price \emph{P} by a geometric Brownian
motion, then $P_{t}$ (the asset value at time $t$) obtains only
positive values. In this case, the additive return is defined as the
difference between the asset's value at time $t$ and its initial
value, i.e., $g_{t}=P_{t}-P_{0}$. Obviously, the additive return
can obtain both positive and negative values (for any time $t>0$),
which is consistent with our requirement that $g_{t}$ be a gamble.
Specifically, in the case of a geometric Brownian motion, $dg_{t}\equiv dp_{t}=p_{t}\cdot\mu\cdot dt+p_{t}\cdot\sigma\cdot dW$,
which implies that $\mu_{t}=p_{t}\cdot\mu$ and $\sigma_{t}=p_{t}\cdot\sigma$
as in Equation (\ref{eq:Wiener}). }\emph{ }and variants of\emph{ }arithmetic Brownian processes and of
mean-reverting processes that are bounded from below (also known as
Ornstein\textendash Uhlenbeck processes; see, e.g., \citealp[Chapter 5]{merton1992continuous};
\citealp{hull1987pricing}; \citealp{meddahi2004temporal}), such
as \citeauthor{cox1985theory2}'s \citeyearpar{cox1985theory2} process
for modeling interest rate.

\subsection{Adapted Definitions}

We define a \emph{local-risk index} at time zero as a function $Q^{l}:\Gamma\rightarrow\mathbb{R}^{++}$
that assigns to each process $g\in\Gamma$ a positive number, which
is interpreted as the process's short-term riskiness at $t=0$. Given
$g\in\Gamma$ with initial parameters $\mu_{0}$ and $\sigma_{0},$
we define three specific local-risk indices (analogous to the corresponding
definitions in Section \ref{subsec:Risk-Indexes}):
\begin{enumerate}
\item The \emph{variance-to-mean local index} $Q_{VM}^{l}(g)$ is equal
to: 
\[
Q_{VM}^{l}\left(g\right)=\frac{\sigma_{0}^{2}}{\mu_{0}}.
\]
\item The\emph{ inverse Sharpe local index} $Q_{IS}^{l}(g)$ is equal to:
\[
Q_{IS}^{l}\left(g\right)=\frac{\sigma_{0}}{\mu_{0}}.
\]
\item The \emph{standard deviation local index} $Q_{SD}^{l}(g)$ is equal
to: $Q_{SD}^{l}\left(g\right)=\sigma_{0}.$
\end{enumerate}
Given a continuous-time process $g\in\Gamma$, decision function $f$
and agent $\left(u,w\right)\in\mathcal{DM}$, let $f_{g}^{\left(u,w\right)}\left(t\right)\equiv f\left(\left(u,w\right),g_{t}\right)$
be the value of the decision function of agent $\left(u,w\right)$
with respect to an investment in $g$ as a function of the duration
of the investment $t$.

The following definition, which deals with general real-valued functions,
will be useful for defining the concept of consistency of indices
in the continuous-time framework.
\begin{defn}
Let $f,h:\mathbb{R^{+}}\Rightarrow\mathbb{R},$ and assume that $\lim_{t\rightarrow0}\frac{f\left(t\right)}{h\left(t\right)}$
is well defined. We say that $f$ is\emph{ uniformly higher} than
$h$ (around zero) and denote it by $f>>h$ if (1) there exists $\bar{t}>0$,
such that $f\left(t\right)>h\left(t\right)$ for each $t\in\left(0,\bar{t}\right)$
and (2) 
\[
\lim_{t\rightarrow0}\frac{f\left(t\right)}{h\left(t\right)}\neq1.
\]

That is, $f>>h$ if\emph{ $f\left(t\right)$} is strictly higher than
$h\left(t\right)$ for any sufficiently small $t$, and the relative
difference between the two functions does not become negligible (as
measured by the ratio $\frac{f\left(t\right)}{h\left(t\right)}$ not
converging to one) in the limit of $t\rightarrow0$.
\end{defn}
We say that a local-risk index is consistent with a decision function
over continuous returns, if the local-risk index of $g'$ is lower
than the index of $g$ iff all risk-averse agents find $g_{t}$ uniformly
more attractive than $g_{t}'$ with respect to any short-term investment.
\begin{defn}
\label{def:Local-risk-index-contion}Local-risk index $Q^{l}:\Gamma\rightarrow\mathbb{R}^{+}$
is\emph{ consistent} with decision function $f$ over the set of continuous
returns $\Gamma$ if, for each pair of continuous-time processes $g,g'\in\Gamma$
and each agent $\left(u,w\right)\in\mathcal{DM}$, we have that
\begin{align*}
Q^{l}\left(g\right) & >Q^{l}\left(g'\right)\Leftrightarrow f_{g}^{\left(u,w\right)}>>f_{g'}^{\left(u,w\right)}.
\end{align*}
\end{defn}
Note that a consistent risk index is unique up to strictly monotone
transformations.

We say that a risk-aversion index is consistent with a decision function
over continuous-time returns, if the risk-aversion index of agent
$\left(u,w\right)$ is strictly higher than the index of $\left(u',w'\right)$
iff agent $\left(u,w\right)$ finds all gambles uniformly less attractive
than agent $\left(u',w'\right)$.
\begin{defn}
\label{def:Risk-aversion-index-cont}Risk-aversion index $\phi:\mathcal{DM}\longrightarrow\mathbb{R}^{+}$
is \emph{consistent} with decision function $f$ over the set of continuous
returns $\Gamma$ if, for each pair of agents $\left(u,w\right),\left(u',w'\right)\in\mathcal{DM}$
and for each gamble $g\in\Gamma$, we have that
\begin{align*}
\phi\left(u,w\right) & >\phi\left(u',w'\right)\Leftrightarrow f_{g}^{\left(u,w\right)}<<f_{g}^{\left(u',w'\right)}
\end{align*}
 Note that a consistent risk-aversion index is unique up to monotone
transformations.
\end{defn}

\subsection{Main Result}

Our main result shows that in each of the decision functions described
above, all agents have the same ranking over all short-term continuous
returns. Moreover, the rankings are consistent with the three risk
indices presented above, and they are the instantaneous versions of
the corresponding indices in the case of normally distributed gambles
and CARA utilities analyzed in Claim \ref{claim:normal-exponent}.
Finally, we adapt to the present setup the classic result that all
continuous short-term returns induce the same ranking over all agents,
which is consistent with the Arrow\textendash Pratt coefficient of
absolute risk aversion (as in the case of normally distributed gambles
and CARA utilities). Formally:
\begin{thm}
\label{thm:main-result}The following conditions hold over the set
of continuous returns $\Gamma$:
\begin{enumerate}
\item The standard deviation index $Q_{SD}^{l}$ is consistent with the
risk premium function $f_{RP}.$
\item The variance-to-mean index $Q_{VM}^{l}$ is consistent with the capital
allocation function $f_{CA}$.
\item The inverse Sharpe index $Q_{IS}^{l}$ is consistent with the optimal
certainty equivalent function $f_{CE}$.
\item The Arrow\textendash Pratt coefficient of absolute risk aversion $\rho$
is consistent with decision functions $f_{CA}$, $f_{CE}$, and $f_{RP}$.
\end{enumerate}
\end{thm}
\begin{proof}[Sketch of proof; formal proof is presented in Appendix \ref{subsec:Proof-of-Theorem}.]
 The value of an asset with a continuous-time return $g$ after a
sufficiently small time $t$ is represented by a gamble $g_{t}$ for
which the magnitudes of all high moments are small relative to the
magnitude of the second moment. Assuming random variables of this
type allows us to use Taylor expansion to approximate the decision
functions, and to obtain the consistent risk indices.

Recall that $\lim_{t\rightarrow0}\frac{\sigma^{2}\left[g_{t}\right]}{t}=\sigma_{0}^{2}$
and $\lim_{t\rightarrow0}\frac{\mu\left[g_{t}\right]}{t}=\mu_{0}$,
which implies for sufficiently small $t$-s that $\sigma^{2}\left[g_{t}\right]\approx t\cdot\sigma_{0}^{2}$
and $\mu\left[g_{t}\right]\approx t\cdot\mu_{0}$. We begin with a
standard approximation of the risk premium function $f_{RP}$ (see,
e.g., \citealp[Chapter 1]{eeckhoudt2005economic}). Recall that the
risk premium was defined implicitly as
\begin{equation}
\mathbb{E}\left[u\left(w+g_{t}\right)\right]=u\left(w+\mathbb{E}\left[g_{t}\right]+f_{RP}\right).\label{eq:sketch-of-proof}
\end{equation}
A second-order Taylor expansion of the left-hand side around $w+\mathbb{E}[g_{t}]$
yields
\begin{align*}
\mathbb{E}\left[u\left(w+g_{t}\right)\right] & \approx\mathbb{E}\bigg[u\left(w+\mathbb{E}\left[g_{t}\right]\right)+u'(w+\mathbb{E}[g_{t}])(g_{t}-\mathbb{E}[g_{t}])+\frac{1}{2}u''\left(w+\mathbb{E}\left[g_{t}\right]\right)(g_{t}-\mathbb{E}[g_{t}])^{2}\bigg]\\
= & u\left(w+E\left[g\right]\right)+\frac{1}{2}u''\left(w+E\left[g_{t}\right]\right)\sigma^{2}\left[g_{t}\right].
\end{align*}
A first-order Taylor expansion of the right-hand side around $w+E[g]$
yields

\[
u\left(w+\mathbb{E}\left[g_{t}\right]+f_{RP}\right)\approx u\left(w+\mathbb{E}\left[g_{t}\right]\right)+u'\left(w+\mathbb{E}\left[g_{t}\right]\right)\cdot f_{RP}.
\]
Combining these equations and isolating $f_{RP}$ yields

\[
f_{RP}\approx\frac{1}{2}\frac{u''\left(w\right)}{u'\left(w\right)}\sigma^{2}\left[g_{t}\right]\approx\frac{1}{2}\frac{u''\left(w\right)}{u'\left(w\right)}t\cdot\sigma_{0}^{2},
\]
which implies that $Q_{SD}^{l}=\sigma_{0}$ (resp., $\rho$=$-\frac{1}{2}\frac{u''\left(w\right)}{u'\left(w\right)}$)
is a consistent risk index (resp., risk-aversion index) for decision
function $f_{RP}$.

In order to analyze $f_{CA}$, we define $C\left(\alpha\right)=\left(f_{RP}\left(\alpha\cdot g_{t}\right)+\mathbb{E}[\alpha\cdot g_{t}]\right)$
to be the certainty equivalent of investment $\alpha$ in $g_{t}$.
Substituting the value of $f_{RP}$ calculated above we get 

\[
C\left(\alpha\right)\approx\alpha^{2}\frac{1}{2}\frac{u''\left(w\right)}{u'\left(w\right)}\sigma^{2}\left[g_{t}\right]+\alpha\cdot\mathbb{E}[g_{t}].
\]
In order to maximize $C\left(\alpha\right)$, we compare the derivative
to zero, to get

\[
\frac{\partial C\left(\alpha\right)}{\partial\alpha}=0\,\,\Leftrightarrow\,\,\alpha^{*}\cdot\frac{u''\left(w\right)}{u'\left(w\right)}\sigma^{2}\left[g_{t}\right]+\mathbb{E}[g_{t}]=0\,\,\Leftrightarrow\,\,f_{CA}=\alpha^{*}\approx-\frac{\mathbb{E}[g_{t}]}{\frac{u''}{u'}\sigma^{2}\left[g_{t}\right]}\approx-\frac{\mu_{0}}{\frac{u''}{u'}\sigma_{0}^{2}},
\]
which implies that the variance-to-mean index $Q_{VM}^{l}$ (resp.,
the Arrow\textendash Pratt coefficient $\rho$) is a consistent risk
index (resp., risk-aversion index) for decision function $f_{CA}$.

Finally, if we calculate $f_{CE}=C\left(\alpha^{*}\right)=C\left(f_{CA}\right)$
we get
\begin{align*}
f_{CE} & \approx\bigg(\frac{\mathbb{E}[g]}{\frac{u''\left(w\right)}{u'\left(w\right)}\sigma^{2}\left[g_{t}\right]}\bigg)^{2}\frac{1}{2}\frac{u''\left(w\right)}{u'\left(w\right)}\sigma^{2}\left[g_{t}\right]+\frac{\mathbb{E}[g_{t}]}{-\frac{u''\left(w\right)}{u'\left(w\right)}\sigma^{2}\left[g_{t}\right]}\mathbb{E}[g_{t}]\\
 & =\frac{1}{2}\frac{1}{-\frac{u''\left(w\right)}{u'\left(w\right)}}\bigg(\frac{\mathbb{E}[g_{t}]}{\sigma\left[g_{t}\right]}\bigg)^{2}\approx\frac{1}{2}\frac{1}{-\frac{u''\left(w\right)}{u'\left(w\right)}}t\cdot\bigg(\frac{\mu_{0}}{\sigma_{0}}\bigg)^{2},
\end{align*}
 which implies that the inverse Sharp index $Q_{IS}^{l}$ (the Arrow\textendash Pratt
coefficient $\rho$) is a consistent risk (risk aversion) index for
decision function $f_{CE}$.
\end{proof}
\begin{rem}[on why the indices in the continuous-time setup coincide with the
indices in the CARA-normal setup]
The expressions that approximate the various functions in the continuous-time
setup consist of two elements: the coefficient of risk aversion with
respect to the initial wealth level, and a function of the first and
second moment of the ``small'' gamble. In the CARA-normal setup
of Section \ref{sec:Normal-Distributions-=000026}, the risk-aversion
coefficient is constant over all wealth levels, and, thus, it is relevant
also to large gambles. In addition, the only moments that matter to
an agent with CARA utility who invests in a normally distributed gamble
are the first two moments. To see that, recall that for CARA utility
$u$ (with coefficient of risk aversion $\rho)$ and normal gamble
$g$, 
\[
\mathbb{E}[u(w+g)]=\mathbb{E}[1-e^{-\rho\cdot\left(w+g\right)}]=1-e^{-\rho\mathbb{E}[w+g]+0.5\rho^{2}\sigma^{2}\left[g\right]}.
\]
Therefore, it seems plausible that the expressions that represent
the decision functions in the CARA-normal setup, depend only on the
first two moments, and, thus, they coincide with the approximated
decision functions that are relevant for short-term investments in
assets with continuous returns.
\end{rem}

\subsection{Weak Consistency for Acceptance/Rejection}

The case of the the acceptance/rejection function $f_{AR}$ has been
analyzed in \citet{schreiber_acceptance_2013}. As the function $f_{AR}$
has only two feasible values (0 or 1), it cannot admit of consistent
risk indices, as in many cases in which one gamble is riskier than
another, an agent may choose to reject both gambles (and his value
of $f_{AR}$ of both gambles would be zero). Nevertheless, one can
define the milder notion of weak consistency, and show that a corollary
to \citeauthor{schreiber_acceptance_2013}'s \citeyearpar{schreiber_acceptance_2013}
result is that the risk index $Q_{VM}^{l}$ is weakly consistent with
the acceptance/rejection function $f_{AR}$.

A local-risk index is weakly consistent with a decision function over
the set of continuous returns, if each agent chooses a weakly lower
value of his investment decision in gamble $g_{t}$ relative to $g'_{t}$
for a sufficiently small \emph{t} if the local risk of $g$ is strictly
higher than the local risk of $g'$. Formally:
\begin{defn}
Local-risk index $Q^{l}:\Gamma\rightarrow\mathbb{R}^{+}$ is \emph{weakly
consistent} with decision function $f$ over the set $\Gamma$ if
for each agent $\left(u,w\right)\in\mathcal{DM}$ and each pair of
continuous-time processes $g,g'\in\Gamma$, there exists time $\bar{t}$,
such that, for each time $t<\bar{t}$, we have that
\[
Q^{l}\left(g\right)>Q^{l}\left(g'\right)\Rightarrow f\left(\left(u,w\right),g_{t}\right)\leq f\left(\left(u,w\right),g'_{t}\right).
\]
\end{defn}
Note that weak consistency does not restrict the agents' choices when
both gambles have the same local-risk index. As a result, a weakly
consistent risk index is unique only up to weakly monotone transformations;
i.e., if $Q$ is a weakly consistent local-risk index with decision
function $f$ over the set of continuous returns $\Gamma$, then risk
index $Q'$ is consistent with function $f$ over this domain if there
exists a weakly increasing mapping $\theta:$ $Q\left(G\right)\rightarrow Q'\left(G\right)$,
such that $Q'\left(g\right)=\theta\left(Q\left(g\right)\right)$ for
each $g\in\Gamma$. In particular, a constant index is trivially a
weakly consistent local-risk index of any decision function.

We say that a risk-aversion index is weakly consistent with a decision
function over continuous-time returns, if for each short-term return,
an agent chooses a (weakly) higher value for his investment decision
in the asset relative to another agent\textquoteright s decision if
the former agent\textquoteright s risk aversion is smaller. Formally:
\begin{defn}
Risk-aversion index $\phi:\mathcal{DM}\longrightarrow\mathbb{R}^{+}$
is \emph{weakly consistent} with decision function $f$ over the domain
of short-term continuous gambles if for each continuous-time process
$g\in\Gamma$ and each pair of agents $\left(u,w\right),\left(u',w'\right)\in\mathcal{DM}$,
there exists a time $\bar{t}$, such that, for each time $t<\bar{t}$,
we have that
\[
\phi\left(u,w\right)>\phi\left(u',w'\right)\Rightarrow f\left(\left(u,w\right),g_{t}\right)\leq f\left(\left(u',w'\right),g_{t}\right).
\]
\end{defn}
\textit{\emph{The following corollary, which is implied by \citet[Theorems 2.2 \& 3.3]{schreiber_acceptance_2013},
shows that the standard deviation index $Q_{VM}^{l}$ and the Arrow\textendash Pratt
coefficient of absolute risk aversion $\rho$ are weakly consistent
with the acceptance/rejection function $f_{AR}$.}}
\begin{cor}[{Implied by \textit{\emph{\citet[Theorems 2.2 \& 3.3]{schreiber_acceptance_2013}}}}]
The following conditions hold over the domain of continuous short-term
decisions:
\begin{enumerate}
\item The variance-to-mean index $Q_{VM}^{l}$ is weakly consistent with
decision function $f_{AR}$.
\item The Arrow\textendash Pratt coefficient $\rho$ is weakly consistent
with decision function $f_{AR}$.
\end{enumerate}
\end{cor}

\subsection{Continuous-Time Processes with Jumps\label{subsec:Continuous-Time-Processes-with}}

The set of continuous-time gambles $\Gamma$ analyzed in this paper
does not allow for jumps. In what follows we show that the absence
of jumps is necessary for our main result. Specifically, we demonstrate
that risk-averse agents rank continuous-time processes with jumps
differently, even for short-term investments, which rules out the
existence of consistent risk indices. Consider, for example, the acceptance/rejection
function $f_{AR}$ (similar conclusions can be drawn for the other
decision functions). \citet{Hart_2011} observes that there are many
pairs of (discrete-time) gambles that are ranked differently by different
risk-averse agents (see \citealt[Footnote 23]{Hart_2011}). Let $h,\tilde{h}$
be such a pair of gambles. 

Consider the following compound Poisson processes $g,\tilde{g}$,
where each has an initial value of zero. The value of each process
changes only when there is a jump. The jumps arrive randomly with
a rate $\lambda$. In process $g$ (resp., $\tilde{g}$) the size
of each jump is distributed according to $h$ (resp., $\tilde{h}$).
Observe that for sufficiently short times the probability of having
two jumps is negligible, and the decision whether to accept or to
reject a gamble depends only on what may happen after a single jump.
This implies that agents who rank the gambles $h,\tilde{h}$ differently,
would also rank $g_{t},\tilde{g}_{t}$ differently, for any sufficiently
short time $t$. This rules out the existence of a consistent risk
index in this setup.

\section{Conclusion\label{sec:Discussion}}

Our main result is that in four central decision problems all risk-averse
agents have the same (problem-dependent) ranking over short-term investments
in risky assets whose returns evolve continuously, and these rankings
are represented by simple well-known indices of risk. The indices
obtained are the same as in the classic model of CARA utilities and
normally distributed gambles. Each problem relates to a different
dimension of risk, and, thus, its ranking is represented by a different
risk index. Finally, adapting a classic result to the present setup,
we show in all of the decision functions analyzed above, the decisions
of agents are consistent with their Arrow\textendash Pratt coefficients
of risk aversion.

\textit{\emph{The proposed indices in our paper are all based on the
first two moments. This is a result of the known property of continuous
stochastic processes for which higher moments go quickly to zero as
the time parameter goes to zero. Hence, multiple indices of risk that
do use higher moments might coincide with our indices when they are
applied to continuous-time processes and short-term investments. For
instance, \citet{schreiber_acceptance_2013} shows that the index
of \citet{Aumann_Serrano_2008} and that of \citet{Foster_Hart_2009}
(which, in general, both depend on all moments of the gamble) coincide
with the variance-to-mean index $Q_{VM}$ for continuous processes
in the limit of $t\rightarrow0$, and \citet{shorrer_1} shows that
there is a continuum of risk indices (which depend also on higher
moments) that are consistent with acceptance/rejection decisions of
agents with respect to small discrete gambles. Indeed, under the assumption
that returns evolve continuously in time, the only relevant parameters
for measuring risk are the first two moments. Our results can be interpreted
as characterizing a necessary condition for a plausible risk index,
namely, that a plausible risk index (with respect to one of the four
decision functions analyzed in the paper) should depend on the first
two moments in the same way as presented in our main result. We leave
for future research the interesting question of how to choose among
the various risk indices that satisfy this necessary condition. One
possible direction for analyzing this question is the axiomatic approach
applied in \citet{shorrer_1} to acceptance/rejection decisions.}}

\textit{\emph{\bibliographystyle{authordate3}
\bibliography{localrisk}
}}

\appendix

\section{Multiplicative Gambles\label{sec:Multiplicative-Gambles}}

In the main text we followed the recent literature of riskiness (initiated
by \citealp{Aumann_Serrano_2008} and \citealp{Foster_Hart_2009})
and focused on decision problems with regard to additive gambles in
units of dollars. However, in most financial applications, it is common
to describe the returns of an asset in relative terms, namely, percentages
(see, e.g., \citealp{markowitz_59} and \citealp{merton1992continuous}),
as this is the way in which returns are described in practice in exchange
markets. Hence, in this section we show that our results hold also
with regard to multiplicative returns.

In some sense, the difference between multiplicative and additive
returns is only a matter of presentation: if one invests x dollars
in a multiplicative gamble \emph{r}, one's payoff will be $x(1+r)$
dollars, and this is just the same payoff as if one invests in an
additive return of $x\cdot r$ dollars. Nevertheless, we think that
presenting the results for multiplicative returns is important for
two reasons: first, as argued in \citet{schreiber_relative_2014},
each investment might have two different aspects of riskiness, absolute
and relative; given two assets, one of them might be riskier in relative
terms but less risky in absolute terms. Therefore it is worthwhile
to study the difference between multiplicative and additive returns
in our setup. As it turns out, this potential difference vanishes
when focusing on short-term investments and we derive in the multiplicative
setup results analogous to those that we have in the additive setup.
Second, in many situations of decision making under risk, the risk-free
interest rate should be taken into account. Since the risk-free interest
rate is calculated in terms of percentages, it is natural to combine
it in decision problems with relative return, as we do here.

\subsection{Adaptation to the Model}

Let $r_{f}>0$ be the risk-free interest rate available for all investors.
A \emph{multiplicative} \emph{risky asset (multiplicative gamble)}
$r$ is a real-valued random variable with an expectation that is
greater than $r_{f},$ and some negative values greater than $-1,$
i.e., $\boldsymbol{E}\left[r\right]>r_{f}$, $\boldsymbol{P}\left[r<0\right]>0$,
and $r\geq-1$. We interpret $r$ as the per-dollar return of the
asset. Let $\mathcal{R}$ denote the set of all multiplicative risky
assets.

We adapt the definitions of our  decision functions to the case of
multiplicative gambles.
\begin{enumerate}
\item The acceptance function $f_{AR}^{m}:\mathcal{DM}\times\mathcal{R}\rightarrow\left\{ 0,1\right\} $
is given by
\[
f_{AR}^{m}\left(\left(u,w\right),r\right)=\begin{cases}
1 & \mathbb{E}\left[u\left(w\cdot\left(1+r\right)\right)\right]\geq u\left(w\cdot(1+r_{f})\right)\\
0 & \mathbb{E}\left[u\left(w\cdot\left(1+r\right)\right)\right]<u\left(w\cdot(1+r_{f})\right),
\end{cases}
\]
where we consider a situation in which an agent faces a binary choice
between investing his entire wealth in a multiplicative gamble $r$
and investing it in the riskless asset with return $r_{f}$.
\item The capital allocation function $f_{CA}^{m}:\mathcal{DM}\times\mathcal{R}\rightarrow\mathbb{R}^{+}\cup\left\{ \infty\right\} $
is given by
\begin{equation}
f_{CA}^{m}\left(\left(u,w\right),r\right)=\arg\max_{\alpha\in\mathbb{R}^{+}}\mathbb{E}\left[u\big(w\cdot(1+r_{f})+\alpha\cdot w\cdot(r-r_{f})\big)\right];\label{eq:f_CA_mult-1}
\end{equation}
if Equation (\ref{eq:f_CA_mult-1}) does not admit of a solution (i.e.,
$\mathbb{E}\left[u\big(w\cdot(1+r_{f})+\alpha\cdot w\cdot(r-r_{f})\big)\right]$
is increasing for all $\alpha$-s), then we set $f_{CA}^{m}\left(\left(u,w\right),r\right)=\infty$.
This function deals with a situation in which an agent decides on
the optimal share $\alpha\geq0$ of his wealth $w$ to invest in the
multiplicative gamble $r$ (where $\alpha>1$ can be induced by leverage),
where his remaining wealth is invested in the riskless asset.
\item The optimal certainty equivalent function ${\color{purple}f_{CE}^{m}:\mathcal{DM}\times\mathcal{R}\rightarrow\mathbb{R}^{+}\cup\left\{ \infty\right\} }$
is defined implicitly as the unique solution to the equation
\[
u\left(w\cdot\left(1+f_{CE}^{m}\right)\right)=\max_{\alpha\in\mathbb{R}^{+}}\mathbb{E}\left[u\left(w\cdot(1+r_{f})+\alpha\cdot w\cdot(r-r_{f})\right)\right]
\]
\begin{equation}
\equiv\mathbb{E}\left[u\big(w(1+r_{f})+f_{CA}^{m}\left(\left(u,w\right),r\right)\cdot w\cdot(r-r_{f})\big)\right];\label{eq:f_CE_mult-1}
\end{equation}
if Equation (\ref{eq:f_CE_mult-1}) does not admit of a solution (i.e.,
$\mathbb{E}\left[u\big(w\cdot(1+r_{f})+\alpha\cdot w\cdot(r-r_{f})\big)\right]$
is increasing for all $\alpha$-s), then we set $f_{CE}^{m}\left(\left(u,w\right),r\right)=\infty$.
This function describes the rate of a constant return that is equivalent
to investing optimally in a multiplicative gamble $r$, where the
remaining wealth is invested in the riskless asset.
\item The risk-premium function $f_{RP}^{m}:\mathcal{DM}\times\mathcal{R}\rightarrow\mathbb{R}^{-}$
is defined implicitly as the unique solution to the equation
\[
\mathbb{E}\left[u\left(w\cdot\left(1+r\right)\right)\right]=u\left(w\cdot\left(1+\mathbb{E}\left[r\right]+f_{RP}^{m}\right)\right),
\]
where $f_{RP}^{m}$ represents the constant (negative) return that
makes the agent indifferent between investing all his wealth in the
multiplicative gamble $r$ and investing in an asset with a constant
return that is equal to the expectation of $r$ plus $f_{RP}^{m}$.
\end{enumerate}
Let $R_{N}\subseteq\mathcal{R}$ be the set of normally distributed
multiplicative gambles (defined analogously to the definition of $\mathcal{G}_{N}$).
The Arrow\textendash Pratt coefficient of \emph{relative} risk aversion,
denoted by $\varrho:\mathcal{DM}\rightarrow\mathbb{R^{++}}$, is defined
as follows: 
\[
\varrho\left(u,w\right)=\frac{-w\cdot u''\left(w\right)}{u'\left(w\right)}.
\]
We adapt the three indices of risk in the main text to the multiplicative
setup and the existence of a risk-free interest rate. Specifically:
\begin{enumerate}
\item The \emph{variance-to-mean index} $Q_{VM}^{m}\left(r\right)$ is equal
to
\[
Q_{VM}^{m}\left(r\right)=\frac{\sigma^{2}\left[r\right]}{\boldsymbol{E}\left[r\right]-r_{f}},\,\,\,\,\,\,\,\,\,\,\,\,\textrm{where}\,\,\sigma^{2}\left[r\right]\equiv E\left[\left(r-E\left[r\right]\right)^{2}\right].
\]
\item The\emph{ inverse Sharpe index} $Q_{IS}^{m}\left(r\right)$ is equal
to
\[
Q_{IS}^{m}\left(r\right)=\frac{\sigma\left[r\right]}{\boldsymbol{E}\left[r\right]-r_{f}}.
\]
\item The \emph{standard deviation} \emph{index} $Q_{SD}^{m}\left(r\right)$
is equal to $Q_{SD}^{m}\left(r\right)=\sigma\left[r\right].$
\end{enumerate}

\subsection{Adapted Results}

The adaptation of Claim \ref{claim:normal-exponent} and Theorem \ref{thm:main-result}
to multiplicative gambles is as follows. Observe that all the results
remain the same, except that the Arrow\textendash Pratt coefficient
of relative risk aversion replaces the coefficient of absolute risk
aversion.
\begin{claim}
\label{claim: multiplicative}The following conditions hold over the
domain $DM_{CARA}\times R_{N}$:
\begin{enumerate}
\item The standard deviation index $Q_{SD}^{m}$ is consistent with  decision
function $f_{RP}^{m}.$
\item The variance-to-mean index $Q_{VM}^{m}$ is consistent with both the
capital allocation function $f_{CA}^{m}$ and the acceptance/rejection
function $f_{AR}^{m}$.
\item The inverse Sharpe index $Q_{IS}^{m}$ is consistent with the  decision
function $f_{CE}^{m}$.
\item The Arrow\textendash Pratt coefficient of relative risk aversion $\varrho$
is consistent with all four decision functions: $f_{AR}^{m}$, $f_{CA}^{m}$,
$f_{CE}^{m}$, and $f_{RP}^{m}$.
\end{enumerate}
\end{claim}
The proof of Claim \ref{claim: multiplicative} is made analogous
to the corresponding proof in the additive case by using the following
identities (details are omitted for brevity):
\begin{enumerate}
\item $f_{AR}^{m}((u,w),r)\equiv f_{AR}((u,w(1+r_{f})),w(r-r_{f}))$,
\item $f_{CA}^{m}((u,w),r)\equiv f_{CA}((u,w(1+r_{f})),w(r-r_{f}))$,
\item $f_{CE}^{m}((u,w),r)\equiv f_{CE}((u,w(1+r_{f})),w(r-r_{f}))/w$,
and
\item $f_{RP}^{m}((u,w),r)\equiv f_{RP}((u,w(1+r_{f})),w(r-r_{f}))/w$.
\end{enumerate}
Recall that in the continuous-time setup, the decision problems are
parameterized by $t$, which is the investment horizon. Previously,
we assumed that a continuous-time random process \emph{g} represents
the additive return of a financial investment. Now the continuous-time
random process $r$ represents the excess multiplicative return: $r_{t}=(P_{t}-P_{0})/P_{0}$.
We assume that the \textit{compound risk-free interest rate} is $r_{f}$
and hence the riskless return over period $t$ is $r_{f}(t)=e^{\mu_{f}\cdot t}-1$.
The adapted definitions of the risk indexes in the multiplicative
setup for the local risk indices are as follows:
\begin{enumerate}
\item The \emph{variance-to-mean local index} $Q_{VM}^{l,m}(g)$ is equal
to
\[
Q_{VM}^{l,m}\left(g\right)=\frac{\sigma_{0}^{2}}{\mu_{0}-\mu_{f}}.
\]
\item The\emph{ inverse Sharpe local index} $Q_{IS}^{l,m}(g)$ is equal
to
\[
Q_{IS}^{l,m}\left(g\right)=\frac{\sigma_{0}}{\mu_{0}-\mu_{f}}.
\]
\item The \emph{standard deviation local index} $Q_{SD}^{l,m}(g)$ is equal
to $Q_{SD}^{l,m}\left(g\right)=\sigma_{0}.$ 
\end{enumerate}
The analogous result to Theorem \ref{thm:main-result} is as follows.
\begin{thm}
\label{thm:Theorem multiplicative}The following conditions hold over
the domain of continuous short-term decisions with respect to multiplicative
gambles:
\begin{enumerate}
\item The standard deviation index $Q_{SD}^{l,m}$ is consistent with the
risk premium function $f_{RP}^{m}.$
\item The variance-to-mean index $Q_{VM}^{l,m}$ is consistent with the
capital allocation function $f_{CA}^{m}$, and it is weakly consistent
with the acceptance/rejection function $f_{AR}^{m}$.
\item The inverse Sharpe index $Q_{IS}^{l,m}$ is consistent with the 
decision function $f_{CE}^{m}$.
\item The Arrow\textendash Pratt coefficient of relative risk aversion $\varrho$
is consistent with decision functions: $f_{CA}^{m}$, $f_{CE}^{m}$,
and $f_{RP}^{m}$, and it is weakly consistent with $f_{AR}^{m}$.
\end{enumerate}
\end{thm}
The proof of Theorem \ref{thm:Theorem multiplicative} is made analogous
to the corresponding proof in the additive case by using the following
identities (details are omitted for brevity):
\begin{enumerate}
\item $f_{AR}^{m}((u,w),r_{t})\equiv f_{AR}((u,w(1+r_{f}(t))),w(r_{t}-r_{f}(t)))$,
\item $f_{CA}^{m}((u,w),r_{t})\equiv f_{CA}((u,w(1+r_{f}(t))),w(r_{t}-r_{f}(t)))$,
\item $f_{CE}^{m}((u,w),r_{t})\equiv f_{CE}((u,w(1+r_{f}(t))),w(r_{t}-r_{f}(t)))/w$,
and
\item $f_{RP}^{m}((u,w),r_{t})\equiv f_{RP}((u,w(1+r_{f}(t))),w(r_{t}-r_{f}(t)))/w$.
\end{enumerate}

\section{Proofs\label{sec:Proofs}}

\subsection{Proof of Claim \ref{claim:normal-exponent}\label{subsec:Proof-of-Claim}}

The following well-known fact, which describes the expectation of
a log-normal distribution, will be useful in our proofs (the standard
proof, which relies on the\emph{ }Laplace transform of the normal
distribution, is omitted for brevity; see, e.g., \citealp[page 132]{forbes2011statistical}).
\begin{fact}
\label{lem:norm-dist}If $y$ is normally distributed with expectation
$\mu$ and standard deviation $\sigma$, then $\mathbb{E}\left[e^{y}\right]=e^{\mu+0.5\sigma^{2}}.$
\end{fact}
Next, we prove Claim \ref{claim:normal-exponent}. Let $g$ be a normally
distributed random variable with expectation $\mu$ and standard deviation
$\sigma$. Let $u$ be a CARA utility with parameter $\rho$, i.e.,
$u\left(x\right)=1-e^{-\rho\cdot x}$. Let $w$ be the arbitrary initial
wealth.
\begin{enumerate}
\item $Q_{SD}$ and $\rho$ are consistent with $f_{RP}.$ The risk premium
$x$ is defined implicitly by 
\begin{eqnarray*}
\mathbb{E}\left[u\left(w+g\right)\right]=\mathbb{E}\left[1-e^{-\rho(w+g)}\right] & = & u\left(w+E\left[g\right]+x\right)=1-e^{-\rho(w+\mu+x)}
\end{eqnarray*}
\[
\Leftrightarrow1-\mathbb{E}\left[e^{-\rho(w+g)}\right]=1-e^{-\rho(w+\mu+x)}\Leftrightarrow\mathbb{E}\left[e^{-\rho(w+g)}\right]=e^{-\rho(w+\mu+x)}.
\]
 By Fact \ref{lem:norm-dist}
\[
\mathbb{E}\left[e^{-\rho(w+g)}\right]=e^{-\rho(w+\mu)+0.5\rho^{2}\sigma^{2}},
\]
which implies 
\[
e^{-\rho(w+\mu)+0.5\rho^{2}\sigma^{2}}=e^{-\rho(w+\mu+x)}\Leftrightarrow-\rho(w+\mu)+0.5\rho^{2}\sigma^{2}=-\rho(w+\mu+x)\Leftrightarrow x=0.5\rho\sigma^{2}.
\]
Thus, $f_{RP}\left(\left(u,w\right),g\right)=0.5\rho\sigma^{2}$,
which implies that $Q_{SD}=\sigma$ is a consistent risk index (and
that $\rho$ is a consistent risk-aversion index with respect to $f_{RP}.$
\item $Q_{VM}$ and $\rho$ are consistent with $f_{AR}$. The agent accepts
the gamble iff 
\[
\mathbb{E}\left(1-e^{-\rho(w+g)}\right)>\mathbb{E}\left(1-e^{-\rho w}\right)
\]
which, by Fact \ref{lem:norm-dist} is equivalent to 
\[
e^{-\rho(w+\mu)+0.5\rho^{2}\sigma^{2}}<e^{-\rho w}\Leftrightarrow0.5\rho<\frac{\mu}{\sigma^{2}}.
\]
Thus, $f_{AR}\left(\left(u,w\right),g\right)=\boldsymbol{1}_{\left\{ 0.5\rho<\frac{\mu}{\sigma^{2}}\right\} }$,
which implies that $Q_{VM}=\frac{\sigma^{2}}{\mu}$ is a consistent
risk index (and that $\rho$ is a consistent risk-aversion index)
with respect to $f_{AR}.$
\item $Q_{VM}$ and $\rho$ are consistent with $f_{CA}$. The capital allocation
function is given by 
\[
f_{CA}\left(\left(u,w\right),g\right)=\arg\max_{\alpha\in\mathbb{R}^{+}}\mathbb{E}\left[u\big(w+\alpha g\big)\right]=\arg\max_{\alpha\in\mathbb{R}^{+}}\mathbb{E}\left[1-e^{-\rho\big(w+\alpha g\big)}\right]\text{.}
\]
It follows from Fact\textbf{ }\ref{lem:norm-dist} that the r.h.s.
of the above equation is equivalent to 
\[
\arg\max_{\alpha\in\mathbb{R}^{+}}\mathbb{E}\left[1-e^{-\rho\big(w+\alpha g\big)}\right]=\arg\max_{\alpha\in\mathbb{R}^{+}}\left(1-e^{-\rho w-\rho\alpha\mu+0.5\rho^{2}\alpha^{2}\sigma}\right)
\]
\[
=\arg\min{}_{\alpha\in\mathbb{R}^{+}}\left(-\rho\alpha\mu+0.5\rho^{2}\alpha^{2}\sigma\right).
\]
The first-order condition is
\[
-\mu+\rho\alpha^{*}\sigma=0\Leftrightarrow\alpha^{*}=\frac{1}{\rho}\frac{\mu}{\sigma^{2}}.
\]
Thus, $f_{CA}\left(\left(u,w\right),g\right)=\frac{1}{\rho}\frac{\mu}{\sigma^{2}}$,
which implies that $Q_{VM}=\frac{\sigma^{2}}{\mu}$ is a consistent
risk index (and that $\rho$ is a consistent risk-aversion index)
with respect to $f_{CA}.$
\item $Q_{IS}$ and $\rho$ are consistent with $f_{CE}$. The optimal certainty
equivalent function is given by
\[
1-e^{-\rho\left(w+f_{CE}\right)}=u\left(w+f_{CE}\right)=\mathbb{E}\left[u\big(w+f_{CA}\left(\left(u,w\right),g\right)\cdot g\big)\right]
\]
\[
\mathbb{=E}\left[u\big(w+\frac{1}{\rho}\frac{\mu}{\sigma^{2}}\cdot g\big)\right]=\mathbb{E}\left[1-e^{-\rho\cdot\big(w+\frac{1}{\rho}\frac{\mu}{\sigma^{2}}\cdot g\big)}\right]=1-e^{-\rho w-\frac{\mu^{2}}{\sigma^{2}}+0.5\frac{\mu^{2}}{\sigma^{2}}},
\]
where the last equality uses Fact \ref{lem:norm-dist}. This implies
that 
\[
1-e^{-\rho\left(w+f_{CE}\right)}=1-e^{-\rho w-\frac{\mu^{2}}{\sigma^{2}}+0.5\frac{\mu^{2}}{\sigma^{2}}}\Leftrightarrow-\rho\left(w+f_{CE}\right)=-\rho w-\frac{\mu^{2}}{\sigma^{2}}+0.5\frac{\mu^{2}}{\sigma^{2}}\Leftrightarrow f_{CE}=\frac{1}{2\rho}\frac{\mu^{2}}{\sigma^{2}}.
\]
 Thus, $f_{CE}\left(\left(u,w\right),g\right)=\frac{1}{2\rho}\frac{\mu^{2}}{\sigma^{2}}$,
which implies that $Q_{IS}=\frac{\sigma}{\mu}$ is a consistent risk
index (resp., $\rho$ is a consistent risk-aversion index) with respect
to $f_{CE}.$
\end{enumerate}

\subsection{\label{subsec:Proof-of-Theorem}Proof of Theorem \ref{thm:main-result}}

The following three lemmas will be useful in our proofs. The first
lemma is a simple version of Ito's well-known lemma (see, e.g., \citealp[Equation 4.4.24]{shreve2004stochastic}).
\begin{lem}[Ito's lemma]
\label{lem:-Let-Ito} Let $s\left(t\right)$ be a random process
described by $ds_{t}=\mu_{t}dt+\sigma_{t}dW$. Let $f(t,s)$ be a
twice-differentiable function; then 
\[
df=\left(\mu_{t}\frac{\partial f}{\partial s}+0.5\sigma_{t}^{2}\frac{\partial^{2}f}{\partial s^{2}}+\frac{\partial f}{\partial t}\right)dt+\frac{\partial f}{\partial s}\sigma_{t}dW.
\]
\end{lem}
The next two lemmas are standard calculus results.
\begin{lem}
\label{lem:standard-calculus}Let $F_{t}(y)$ be a set of real-valued,
continuous, and weakly increasing functions, with $0<t\leq T$ and
$y\in\mathbb{R}$. Assume that there exists a continuous and strictly
increasing function $F(y)$ such that (1) $\forall y,F(y)=\lim_{t\rightarrow0}F_{t}(y)$,
and (2) $\exists y^{*}$, s.t. $F(y^{*})=0$. Then, there exists $\bar{t}>0$
s.t. 
\[
\forall t<\bar{t}\ \exists y_{t}\text{ s.t. }F_{t}(y_{t})=0,\,\,\,\textrm{and}\,\,\,\lim_{t\rightarrow0}{y_{t}}=y^{*}.
\]
\end{lem}
\begin{proof}
Let $\delta>0$. We have to show that there exists $\bar{t}$ s.t.
$\forall t<\bar{t}$ there is a value $y_{t}$ satisfying $|y_{t}-y^{*}|<\delta$
and $F_{t}(y_{t})=0$. Since $F(y)$ is strictly increasing there
exists a positive number $C$ such that $F(y^{*}-\delta)<-C$ and
$F(y^{*}+\delta)>C$. Condition ($1$) implies that there exists $\bar{t}$
s.t. $\forall t<\bar{t}$, 
\[
|F_{t}(y^{*}+\delta)-F(y^{*}+\delta)|<C,\,\,\,\textrm{and}\,\,\,\,\,|F_{t}(y^{*}-\delta)-F(y^{*}-\delta)|<C.
\]
Hence, $F_{t}(y^{*}-\delta)<0$ and $F_{t}(y^{*}+\delta)>0$. Since
$F_{t}$ is continuous, $\exists y_{t}\in(y^{*}-\delta,y^{*}+\delta)$
s.t. $F_{t}(y_{t})=0$.
\end{proof}
\begin{lem}
\label{lem:200} Let $F_{t}(\alpha)$ be a set of twice-differentiable
strictly concave functions where $0<t\leq T$ and $\alpha\in\mathbb{R}$,
and let $F$ be a twice-differentiable strictly concave function such
that (1) $\forall\alpha,\ F(\alpha)=\lim_{t\rightarrow0}{F_{t}(\alpha)},$
and (2) $\exists\alpha^{*}\in\mathbb{R}$ such that $\alpha^{*}=\arg\max_{\alpha\in\mathbb{R}}{F(\alpha)}$.
Then, there exists $\bar{t}>0$ such that 
\[
\forall t<\bar{t},\exists\alpha_{t}\in\mathbb{R}\text{ s.t. }\alpha_{t}=\arg\max_{\alpha}F_{t}(\alpha),\,\,\,\textrm{and}\,\,\,\lim_{t\rightarrow0}{\alpha_{t}}=\alpha^{*}.
\]
\end{lem}
\begin{proof}
We have to show that, given $\delta>0$, there exists $\bar{t}>0$
such that $\forall t<\bar{t}$, $\exists\alpha_{t}$, which maximizes
$F_{t}(\alpha)$, and that $|\alpha_{t}-\alpha^{*}|<\delta$. Let
$\delta_{1}=\min\{F(\alpha^{*})-F(\alpha^{*}-\delta),F(\alpha^{*})-F(\alpha^{*}+\delta)\}$.
There exists $\bar{t}$ s.t. $\forall t<\bar{t}$, 
\[
|F_{t}(\alpha^{*})-F(\alpha^{*})|<\delta_{1}/3,\,\,\,\,\,|F_{t}(\alpha^{*}+\delta)-F(\alpha^{*}+\delta)|<\delta_{1}/3,\,\,\,\textrm{and}
\]
\[
|F_{t}(\alpha^{*}-\delta)-F(\alpha^{*}-\delta)|<\delta_{1}/3.
\]
Hence, $\forall t<\bar{t}$, 
\[
F_{t}(\alpha^{*})>F_{t}(\alpha^{*}-\delta)\text{ and }F_{t}(\alpha^{*})>F_{t}(\alpha^{*}+\delta).
\]
Since for all $t$, $F_{t}$ is weakly concave, there exists $\alpha_{t}\in(\alpha^{*}-\delta,\alpha^{*}-\delta)$,
which is the argmax of $F_{t}$.
\end{proof}
Next, we prove the main theorem. Let $g\in\Gamma$ be a continuous-time
random process, and let $\left(u,w\right)\in\mathcal{DM}$ be a decision
maker.
\begin{enumerate}
\item $Q_{SD}^{l}$ and $\rho$ are consistent with $f_{RP}.$ For every
$t>0$, let $F_{t}$ be  defined as follows: 
\[
F_{t}(x)=\frac{u\big(w+\mathbb{E}\left[g_{t}\right]+x\cdot t\big)-\mathbb{E}\left[u\big(w+g_{t}\big)\right]}{t}.
\]
By definition, if for some value of $x$, $F_{t}(x)=0$, then $x\cdot t=f_{RP}(\left(u,w\right),g_{t})$.
To calculate the limit of $F_{t}$ as $t$ goes to zero, it is simpler
to look at $F_{t}$ as the difference between two functions $k_{t}$
and $h_{t}$, defined by 
\[
k_{t}(x)=\frac{u(w+\mathbb{E}\left[g_{t}\right]+x\cdot t)-u(w)}{t},\,\,\,\,\textrm{and}
\]
\begin{equation}
h_{t}=\frac{\mathbb{E}\left[u\bigg(w+g_{t}\bigg)\right]-u(w)}{t}.\label{eq:h_t-1}
\end{equation}
 Clearly, 
\[
F_{t}(x)=k_{t}(x)-h_{t}
\]
for every value of x. The limit of $k_{t}(x)$ as $t$ goes to zero
is simply the derivative with respect to $t$ at $w$:
\begin{equation}
\lim_{t\rightarrow0}k_{t}(x)=u'(w)\cdot(\mu_{0}+x).\label{eq:kt}
\end{equation}
By applying Ito's lemma
\[
h_{t}=\frac{\mathbb{E}\left[\int_{0}^{t}\left(\mu_{q}u'_{q}+\frac{1}{2}\sigma_{q}^{2}u''_{q}\right)dq\right]}{t}+\frac{\mathbb{E}\left[\int u'_{q}\sigma_{q}dW\right]}{t},
\]
where $u'_{q}\equiv du(w_{q})/d(w_{q})$, $u''_{q}\equiv du^{2}(w_{q})/d^{2}(w_{q})$,
and $w_{q}=w+g_{q}$. Since we assumed that $g$ is bounded from below,
the concavity and monotonicity of $u$ implies that $u'_{q}$ is bounded.
In addition, we assumed that the $\sigma_{t}$ satisfies the square-integrability
condition and, therefore, that $E\bigg[\int_{0}^{t}\sigma_{q}^{2}dq\bigg]$
is finite. These two assumptions imply that $E\bigg[\int_{0}^{t}(u'_{q}\sigma_{q})^{2}dq\bigg]$
is finite and, therefore, that $\int_{0}^{t}u'_{q}\sigma_{t}dW$ is
a martingale; see \citet[Theorem 4.3.1. on page 134]{shreve2004stochastic}.
Hence, $h_{t}$ can be rewritten as follows:
\[
h_{t}=\frac{\mathbb{E}\left[\int_{0}^{t}\left(\mu_{q}u'_{q}+\frac{1}{2}\sigma_{q}^{2}u''_{q}\right)dq\right]}{t}.
\]
Since $\mu_{q}$, $\sigma_{q},u'_{q}$ and $u''_{q}$ are all continuous,
according to the mean-value theorem for integration, for each realization
of $g$ there exists some $x\in\left(0,t\right)$ for which
\[
\frac{\int_{0}^{t}\left(\mu_{q}u'_{q}+\frac{1}{2}\sigma_{q}^{2}u''_{q}\right)dq}{t}=\mu_{x}u'_{x}+\frac{1}{2}\sigma_{x}^{2}u''_{x}.
\]
As t goes to zero this expression converges to $\mu_{0}u'(w)+\frac{1}{2}\sigma_{0}^{2}u''(w)$.
Since for every realization of $g$ it converges to the exact same
number, the expectation of this expression also converges to this
number. Therefore,
\begin{equation}
\lim_{t\rightarrow0}{h_{t}}=\mu_{0}u'(w)+\frac{1}{2}\sigma_{0}^{2}u''(w).\label{eq:ht}
\end{equation}
It follows from Equations (\ref{eq:kt}) and (\ref{eq:ht}) that
\[
F(x)\equiv\lim_{t\rightarrow0}{F_{t}(x)}=u'(w)x-\frac{1}{2}\sigma_{0}^{2}u''(w).
\]
Let $x^{*}$ be the real number s.t. $F(x^{*})=0$, i.e.,
\[
x^{*}=\frac{1}{2}\frac{u''(w)}{u'(w)}\sigma_{0}^{2}.
\]
It is easy to see that the two conditions of Lemma \ref{lem:standard-calculus}
are satisfied: for all \emph{t}, first $F_{t}$ is continuous as it
is the sum of continuous functions, and second, $F_{t}$ is a strictly
increasing function since $u$ is an increasing function. It follows
from the lemma that there exists $\bar{t}$ and $x_{t}$ such that
$F_{t}\left(x_{t}\right)=0$ for each $t<\bar{t}$, and
\[
\lim_{t\rightarrow0}x_{t}=x^{*},
\]
where, by definition, $f_{RP}(\left(u,w\right)\text{,}g_{t})=x_{t}t$.
Note that since $\frac{u''(w)}{u'(w)}$ is negative, $x^{*}$ is negative
as well, and therefore $x^{*}$ (and $x^{*}\cdot t$ for all $t>0$)
is strictly decreasing with $\rho=-\frac{u''(w)}{u'(w)}$ and with
$Q_{SD}^{l}=\sigma_{0}$. \\
Next, we would like to show that $Q_{SD}^{l}$ and $\rho$ are consistent
with $f_{RP}.$ We begin by showing that $Q_{SD}^{l}(g)>Q_{SD}^{l}(g')$
implies that $\left(f_{RP}\right)_{g}^{\left(u,w\right)}<<\left(f_{RP}\right)_{g'}^{\left(u,w\right)}$
for any $\left(u,w\right)\in\mathcal{DM}$. Fix a decision maker $\left(u,w\right)$,
and let $x^{*}(g)\equiv x^{*}\left(\left(u,w\right),g\right)$ (and
use a similar notation for $x_{t}(g)$). Let $g,g'\in\Gamma$ be two
processes satisfying $Q_{SD}^{l}(g)>Q_{SD}^{l}(g')$. Then $x^{*}(g)<x^{*}(g')$,
and from the fact that $x_{t}\rightarrow x^{*}$it follows that there
exists $\bar{t}>0$, such that for each $t\in\left(0,\bar{t}\right),$
$x_{t}(g)\cdot t<x_{t}(g')\cdot t$, which implies that $f_{RP}\left(\left(u,w\right),g_{t}\right)<f_{RP}\left(\left(u,w\right),g_{t}'\right)$.
In addition, 
\[
lim_{t\rightarrow0}\frac{f_{RP}\left(\left(u,w\right),g_{t}\right)}{f_{RP}\left(\left(u,w\right),g'_{t}\right)}=lim_{t\rightarrow0}\frac{x_{t}\left(g_{t}\right)t}{x_{t}\left(g_{t}'\right)t}=\frac{x^{*}\left(g\right)}{x^{*}\left(g'\right)}=\frac{\left(Q_{SD}^{l}(g)\right)^{2}}{\left(Q_{SD}^{l}(g')\right)^{2}}\neq1,
\]
which proves that $\left(f_{RP}\right)_{g}^{\left(u,w\right)}<<\left(f_{RP}\right)_{g'}^{\left(u,w\right)}$.
Similarly, we show that $\rho\left(u',w'\right)>\rho\left(u'',w''\right)$
implies that $\left(f_{RP}\right)_{g}^{\left(u',w'\right)}<<\left(f_{RP}\right)_{g}^{\left(u'',w''\right)}$
for any $g\in\Gamma$. Fix a process $g\in\Gamma$, and let $x^{*}(u,w)\equiv x^{*}\left(\left(u,w\right),g\right)$
(and use a similar notation for $x_{t}(u,w)$). Let $\left(u',w'\right),\left(u'',w''\right)\in\mathcal{DM}$
be two agents satisfying $\rho\left(u',w'\right)>\rho\left(u'',w''\right)$.
Then $x^{*}\left(u',w'\right)<x^{*}\left(u'',w''\right)$, and from
the fact that $x_{t}\rightarrow x^{*}$ it follows that there exists
$\bar{t}>0$, such that for each $t\in\left(0,\bar{t}\right)$, $x_{t}(u'',w'')\cdot t<x_{t}(u',w')\cdot t$,
implying that $f_{RP}\left(\left(u'',w''\right),g_{t}\right)<f_{RP}\left(\left(u',w'\right),g_{t}\right)$.
In addition, 
\[
lim_{t\rightarrow0}\frac{f_{RP}\left(\left(u'',w''\right),g_{t}\right)}{f_{RP}\left(\left(u',w'\right),g_{t}\right)}=lim_{t\rightarrow0}\frac{x_{t}\left(u'',w''\right)t}{x_{t}\left(u',w'\right)t}=\frac{x^{*}\left(u'',w''\right)}{x^{*}\left(u',w'\right)}=\frac{\rho\left(u'',w''\right)}{\rho\left(u',w'\right)}\neq1,
\]
which proves that $\left(f_{RP}\right)_{g}^{\left(u',w'\right)}<<\left(f_{RP}\right)_{g}^{\left(u'',w''\right)}$.\\
For the other direction, given some agent $\left(u,w\right)$, if
for two processes $g$ and $g'$ there is some $\bar{t}$ s.t. $f_{RP}\left(\left(u,w\right),g_{t}\right)<f_{RP}\left(\left(u,w\right),g'_{t}\right)$
for every $0<t<\bar{t}$, and the ratio $f_{RP}\left(\left(u,w\right),g_{t}\right)/f_{RP}\left(\left(u,w\right),g_{t}'\right)$
does not go to 1 when \emph{t} goes to zero, then $x_{t}<x'_{t}$
for all $t<\bar{t}$, implying that the limits also satisfy $x^{*}<x'^{*}$
and, therefore, $\sigma_{0}>\sigma_{0}'$. Similarly, given some process
$g$, if for two agents $(u',w')$ and $(u'',w'')$ there is some
$\bar{t}$ s.t. $f_{RP}\left(\left(u',w'\right),g_{t}\right)<f_{RP}\left(\left(u'',w''\right),g_{t}\right)$
for every $0<t<\bar{t}$, and the ratio $f_{RP}\left(\left(u,w\right),g_{t}\right)/f_{RP}\left(\left(u',w'\right),g_{t}\right)$
does not go to 1 when \emph{t} goes to zero, then $x^{*}<x'^{*}$,
implying that $\left(u',w'\right)$ is locally more averse to risk
than $\left(u'',w''\right)$.
\item $Q_{VM}^{l}$ and $\rho$ are consistent with $f_{CA}$. The capital
allocation function is defined by 
\[
f_{CA}\left(\left(u,w\right),g_{t}\right)=\arg\max_{\alpha\in\mathbb{R}^{+}}\mathbb{E}\left[u\big(w+\alpha\cdot g_{t}\big)\right],
\]
where $f_{CA}\left(\left(u,w\right),g_{t}\right)$ equals infinity
if there is no internal solution. For every $t>0$, let $F_{t}$ be
the function defined as follows: 
\begin{eqnarray}
F_{t}(\alpha)=\frac{\mathbb{E}\left[u\bigg(w+\alpha g_{t}\bigg)\right]-u(w)}{t}.\label{eq:F_t_CA}
\end{eqnarray}
By Ito's lemma, 
\[
F_{t}(\alpha)=\frac{\mathbb{E}_{0}\left[\int_{0}^{t}\alpha\mu_{q}u'_{q}+\frac{1}{2}\alpha^{2}\sigma_{q}^{2}u''_{q}\ dq\right]}{t}+\frac{\mathbb{E}\bigg[\int_{0}^{t}\alpha u'_{q}\sigma_{q}dW\bigg]}{t},
\]
where $u'_{q}\equiv du(w_{q})/d(w_{q})$, $u''_{q}\equiv du^{2}(w_{q})/d^{2}(w_{q})$,
and $w_{q}=w+\alpha g_{q}$. For the same reason as in the case of
$f_{RP}$, the expression on the right-hand side $\mathbb{E}\bigg[\int_{0}^{t}\alpha u'_{q}\sigma_{q}dW\bigg]$
is zero and therefore it can be omitted. \\
We define $F(\alpha)$ to be the limit of $F_{t}(\alpha)$ as $t$
goes to zero. For the same reason as in the case of $f_{RP}$ it equals
to:
\begin{equation}
F(\alpha)\equiv\lim_{t\rightarrow0}{F_{t}(\alpha)}=\alpha\mu_{0}u'(w)+\frac{1}{2}\alpha^{2}\sigma_{0}^{2}u''(w).\label{eq:F_t_CA2}
\end{equation}
We denote by $\alpha^{*}$ the value of $\alpha$ that maximizes $F\left(\alpha\right)$:
\begin{equation}
\alpha^{*}=\arg\max_{\alpha}{F(\alpha)}=-\frac{u'(w)}{u''(w)}\frac{\mu_{0}}{\sigma_{0}^{2}}.\label{eq:alpha-star}
\end{equation}
The two conditions of Lemma \ref{lem:200} are satisfied: first, by
definition, the limit of $F_{t}$ is $F$. Second, we represent $F_{t}$
as the sum of two expressions
\[
F_{t}(\alpha)=\frac{\alpha\cdot\mathbb{E}_{o}\left[\int_{0}^{t}\mu_{q}u'_{q}dq\right]}{t}+\frac{\alpha^{2}\cdot\mathbb{E}_{o}\left[\frac{1}{2}\int_{0}^{t}\sigma_{q}^{2}u''_{q}\ dq\right]}{t}.
\]
Since we assume that $u''$ is negative, $F_{t}$ is strictly concave
with $\alpha$ and the second condition of the lemma is satisfied.
\\
By the lemma, there exists $\bar{t}>0$ such that $\alpha_{t}$ maximizes
$F_{t}$ for all $t<\bar{t}$, and\footnote{\label{fn:The-analysis-implies}The analysis implies that $g_{t}$
is a ``gamble'' for each $t<\bar{t}$. To see that, note that we
have shown that for every process $g$, and for every strictly concave
utility function, there exists $\bar{t}$ such that for every $t<\bar{t}$,
the solution of the maximization problem is internal. This implies
that for every such $t$, $\boldsymbol{E}[g_{t}]>0$. Otherwise, a
risk-averse agent would be better off by choosing $\alpha_{t}=0$,
contradicting our result here that $\alpha_{t}>0$ for a sufficiently
short time $t$. Similarly, the analysis implies that $P(g_{t}<0)>0$
for a sufficiently short time $t$. Otherwise, for every $\alpha_{t}$
and $\epsilon>0$, $\left(\alpha_{t.}+\epsilon\right)g_{t}$ would
first-order stochastically dominate $\alpha_{t}g_{t}$ and therefore
any agent would be better off enlarging any given $\alpha_{t.}$,
which implies that the solution is not internal, contradicting our
result that some finite $\alpha_{t}>0$ maximizes $F_{t}$. These
two properties of $g_{t}$ imply that $g_{t}$ is a gamble.}
\[
\lim_{t\rightarrow0}\alpha_{t}=\alpha^{*}.
\]
Note that the limit $\alpha^{*}$ is strictly decreasing with $\rho=-u''(w)/u'(w)$,
and with $Q_{VM}^{l}=\sigma_{0}^{2}/\mu_{0}$. \\
Next we would like to show that $Q_{VM}^{l}$ and $\rho$ are consistent
with $f_{CA}$, where by definition $f_{CA}\left(\left(u,w\right),g_{t}\right)=\left(f_{CA}\right)_{g}^{\left(u,w\right)}(t)=\alpha_{t}(\left(u,w\right),g_{t}).$
For the first direction, we have to show that if for two processes
$g,g'\in\Gamma$ , $Q_{VM}^{l}(g)>Q_{VM}^{l}(g')$, then $\left(f_{CA}\right)_{g}^{\left(u,w\right)}<<\left(f_{CA}\right)_{g'}^{\left(u,w\right)}$.
Indeed, $Q_{VM}^{l}(g)>Q_{VM}^{l}(g')$ implies that $\alpha^{*}(g)<\alpha^{*}(g')$,
and from the convergence of $\alpha_{t}$ it follows that there exists
$\bar{t}>0$, such that for each $t\in\left(0,\bar{t}\right)$, $\alpha_{t}(g)<\alpha_{t}(g')$.
Since $\alpha^{*}(g)<\alpha^{*}(g')$, it follows that $\lim_{t\rightarrow0}\alpha_{t}(g)/\alpha_{t}(g')\neq1$
and therefore that $\left(f_{CA}\right)_{g}^{\left(u,w\right)}<<\left(f_{CA}\right)_{g'}^{\left(u,w\right)}$.
Similarly, if for two agents $\left(u',w'\right)$ and $\left(u'',w''\right)$,
$\rho\left(u',w'\right)>\rho\left(u'',w''\right)$, then $\alpha^{*}(u',w')<\alpha^{*}(u'',w'')$,
and from the convergence of $\alpha_{t}$ it follows that there exists
some $\bar{t}>0$, such that for each $t\in\left(0,\bar{t}\right)$,
$\alpha_{t}(u',w')<\alpha_{t}(u'',w'')$. Since $\alpha^{*}$ (the
limit of $\alpha_{t}$) is positive, it follows that $\lim_{t\rightarrow0}\alpha_{t}(u',w')/\alpha(u'',w'')\neq1$
and therefore that $\left(f_{CA}\right)_{g}^{\left(u',w'\right)}<<\left(f_{CA}\right)_{g}^{\left(u'',w''\right)}$.\\
For the other direction, let $g,g'\in\Gamma$ be two processes for
which, for any decision maker $\left(u,w\right)$, $\left(f_{CA}\right)_{g}^{\left(u,w\right)}<<\left(f_{CA}\right)_{g'}^{\left(u,w\right)}$.
Indeed, the limits of $f_{CA}\left(\left(u,w\right),g_{t}\right)$
and $f_{CA}\left(\left(u,w\right),g'_{t}\right)$ when $t$ goes to
zero satisfy $\alpha^{*}(g')>\alpha^{*}(g)$ and, therefore, $Q_{VM}^{l}(g)>Q_{VM}^{l}(g')$.
Similarly, let $(u,w)$ and $(u',w')$ be two decision makers for
which $\left(f_{CA}\right)_{g}^{\left(u',w'\right)}<<\left(f_{CA}\right)_{g}^{\left(u'',w''\right)}$.
Indeed, the limits of $f_{CA}\left(\left(u,w\right),g_{t}\right)$
and $f_{CA}\left(\left(u,w\right),g'_{t}\right)$ when $t$ goes to
zero satisfy $\alpha^{*}(u',w')<\alpha^{*}(u'',w'')$ and, therefore,
$\rho\left(u',w'\right)>\rho\left(u'',w''\right)$.
\item $Q_{IS}^{l}$ and $\rho$ are consistent with $f_{CE}$.  For every
$t>0$, let $F_{t}$ be  defined as follows:
\[
F_{t}(z)=\frac{u(w+z\cdot t)-\mathbb{E}\left[u\bigg(w+\alpha g_{t}\bigg)\right]}{t}.
\]
It is easy to see that if $\alpha$ is the optimal allocation and
$F_{t}(z)=0$ then $z\cdot t=f_{CA}(\left(u,w\right),g_{t})$. To
calculate the limit of $F_{t}$ as $t$ goes to zero, it is simpler
to look at $F_{t}$ as the difference between two functions $k_{t}$
and $h_{t}$, defined by 
\[
k_{t}(z)=\frac{u(w+zt)-u(w)}{t},\,\,\,\,\textrm{and}
\]
\begin{equation}
h_{t}=\frac{\mathbb{E}\left[u\bigg(w+\alpha g_{t}\bigg)\right]-u(w)}{t}.\label{eq:h_t}
\end{equation}
 Clearly, 
\[
F_{t}(z)=k_{t}(z)-h_{t}
\]
for every value of $z$. The limit of $k_{t}(z)$ as $t$ goes to
zero is simply the derivative:
\[
\lim_{t\rightarrow0}k_{t}(z)=u'(w)\cdot z.
\]
Using Ito's lemma, and taking the limit (as we did in Equations \ref{eq:F_t_CA}
and \ref{eq:F_t_CA2}), we get
\[
\lim_{t\rightarrow0}{h_{t}}=\alpha\mu_{0}u'(w)+\frac{1}{2}\alpha^{2}\sigma_{0}^{2}u''(w).
\]
Recall that according to Equation \eqref{eq:alpha-star}, 
\[
\alpha^{*}=-\frac{u'(w)}{u''(w)}\frac{\mu_{0}}{\sigma_{0}^{2}}.
\]
Plugging $\alpha=\alpha^{*}$ into $h_{t}$, we get 
\[
\lim_{t\rightarrow0}h_{t}=-\frac{\left(u'\left(w\right)\right)^{2}\mu_{0}^{2}}{u''\left(w\right)\sigma^{2}}+\frac{1}{2}\frac{\left(u'\left(w\right)\right)^{2}\mu_{0}^{2}}{u''\left(w\right)\sigma^{2}}=-\frac{1}{2}\frac{\left(u'\left(w\right)\right)^{2}\mu_{0}^{2}}{u''\left(w\right)\sigma^{2}}.
\]
We define $F(z)$ to be the limit of $F_{t}(z)$, where $t$ goes
to zero: 
\begin{eqnarray*}
F(z) & \equiv & \lim_{t\rightarrow0}{F_{t}(z)}=\lim_{t\rightarrow0}{k_{t}(z)}-\lim_{t\rightarrow0}{h_{t}(z)}=u'(w)z+\frac{1}{2}\frac{(u_{0}')^{2}}{u_{0}''}\big(\frac{\mu_{0}}{\sigma_{0}}\big)^{2}.
\end{eqnarray*}
We define $z^{*}$ to be the value that results in $F(z^{*})=0$:
\[
z^{*}=-\frac{1}{2}\frac{u'(w)}{u''(w)}\big(\frac{\mu_{0}}{\sigma_{0}}\big)^{2}.
\]
For every $t$, $F_{t}(z)$ is continuous and strictly increasing
satisfying the conditions of Lemma \ref{lem:standard-calculus}, therefore
by the lemma there is $\bar{t}$ such that $F_{t}(z_{t})=0$ for every
$t\in(0,\bar{t})$, implying that $z_{t}\cdot t$ is the certainty
equivalent of the optimal investment in the gamble with horizon $t$,
and that
\[
\lim_{t\rightarrow0}z_{t}=z^{*}.
\]
It is easy to see that $z^{*}$ (and therefore $z^{*}t$ for all $t$)
is strictly decreasing with $\rho=-\frac{u''(w)}{u'(w)}$ and with
$Q_{IS}^{l}=\frac{\sigma}{\mu}$.\\
Next, we would like to show that $Q_{IS}^{l}$ and $\rho$ are consistent
with $f_{CE}$. We begin by showing that $Q_{IS}^{l}(g)>Q_{IS}^{l}(g')$
implies that $\left(f_{CE}\right)_{g}^{\left(u,w\right)}<<\left(f_{CE}\right)_{g'}^{\left(u,w\right)}$
for any $\left(u,w\right)\in\mathcal{DM}$. Fix a decision maker $\left(u,w\right)$,
and let $z^{*}(g)\equiv z^{*}\left(\left(u,w\right),g\right)$ (and
use a similar notation for $z_{t}(g)$). Let $g,g'\in\Gamma$ be two
processes satisfying $Q_{IS}^{l}(g)>Q_{IS}^{l}(g')$. Then $z^{*}(g)<z^{*}(g')$,
and from the fact that $z_{t}\rightarrow z^{*}$ it follows that there
exists $\bar{t}>0$, such that for each $t\in\left(0,\bar{t}\right),$
$z_{t}(g)\cdot t<z_{t}(g')\cdot t$, which implies that $f_{CE}\left(\left(u,w\right),g_{t}\right)<f_{CE}\left(\left(u,w\right),g_{t}'\right)$.
In addition, 
\[
lim_{t\rightarrow0}\frac{f_{CE}\left(\left(u,w\right),g_{t}\right)}{f_{CE}\left(\left(u,w\right),g'_{t}\right)}=lim_{t\rightarrow0}\frac{z_{t}\left(g_{t}\right)t}{z_{t}\left(g_{t}'\right)t}=\frac{z^{*}\left(g\right)}{z^{*}\left(g'\right)}=\frac{\left(Q_{IS}^{l}(g)\right)^{2}}{\left(Q_{IS}^{l}(g')\right)^{2}}\neq1,
\]
which proves that $\left(f_{CE}\right)_{g}^{\left(u,w\right)}<<\left(f_{CE}\right)_{g'}^{\left(u,w\right)}$.
Similarly, we show that $\rho\left(u',w'\right)>\rho\left(u'',w''\right)$
implies that $\left(f_{CE}\right)_{g}^{\left(u',w'\right)}<<\left(f_{CE}\right)_{g}^{\left(u'',w''\right)}$
for any $g\in\Gamma$. Fix a process $g\in\Gamma$, and let $z^{*}(u,w)\equiv z^{*}\left(\left(u,w\right),g\right)$
(and use a similar notation for $z_{t}(u,w)$). Let $\left(u',w'\right),\left(u'',w''\right)\in\mathcal{DM}$
be two agents satisfying $\rho\left(u',w'\right)>\rho\left(u'',w''\right)$.
Then $z^{*}\left(u',w'\right)<z^{*}\left(u'',w''\right)$, and from
the fact that $z_{t}\rightarrow z^{*}$ it follows that there exists
$\bar{t}>0$, such that for each $t\in\left(0,\bar{t}\right)$, $z_{t}(u'',w'')\cdot t<z_{t}(u',w')\cdot t$,
implying that $f_{CE}\left(\left(u'',w''\right),g_{t}\right)<f_{CE}\left(\left(u',w'\right),g_{t}\right)$.
In addition, 
\[
lim_{t\rightarrow0}\frac{f_{CE}\left(\left(u'',w''\right),g_{t}\right)}{f_{CE}\left(\left(u',w'\right),g_{t}\right)}=lim_{t\rightarrow0}\frac{z_{t}\left(u'',w''\right)t}{z_{t}\left(u',w'\right)t}=\frac{z^{*}\left(u'',w''\right)}{z^{*}\left(u',w'\right)}=\frac{\rho\left(u'',w''\right)}{\rho\left(u',w'\right)}\neq1,
\]
which proves that $\left(f_{CE}\right)_{g}^{\left(u',w'\right)}<<\left(f_{CE}\right)_{g}^{\left(u'',w''\right)}$.\\
For the other direction, given some agent $\left(u,w\right)$, if
for two processes $g$ and $g'$ there is some $\bar{t}$ s.t. $f_{CE}\left(\left(u,w\right),g_{t}\right)<f_{CE}\left(\left(u,w\right),g'_{t}\right)$
for every $0<t<\bar{t}$, and the ratio $f_{CE}\left(\left(u,w\right),g_{t}\right)/f_{CE}\left(\left(u,w\right),g_{t}'\right)$
does not go to 1 when \emph{t} goes to zero, then $x_{t}<x'_{t}$
for all $t<\bar{t}$, implying that the limits also satisfy $z^{*}<z'^{*}$
and, therefore, $\sigma_{0}>\sigma_{0}'$. Similarly, given some process
$g$, if for two agents $(u',w')$ and $(u'',w'')$ there is some
$\bar{t}$ such that $f_{CE}\left(\left(u',w'\right),g_{t}\right)<f_{CE}\left(\left(u'',w''\right),g_{t}\right)$
for every $t\in\left(0,\bar{t}\right)$, and the ratio $f_{CE}\left(\left(u'',w''\right),g_{t}\right)/f_{CE}\left(\left(u',w'\right),g_{t}\right)$
does not go to 1 when \emph{t} goes to zero, then $z^{*}<z'^{*}$,
implying that $\left(u',w'\right)$ is locally more averse to risk
than $\left(u'',w''\right)$.
\end{enumerate}

\end{document}